\documentclass[a4paper,12pt]{article}

\usepackage[english]{babel}
\usepackage[utf8x]{inputenc}
\usepackage[T1]{fontenc}

\usepackage[a4paper,top=3cm,bottom=2cm,left=3cm,right=3cm,marginparwidth=1.75cm]{geometry}

\usepackage{multirow}
\usepackage{mathrsfs}
\usepackage{amsfonts}
\usepackage{amsmath}
\usepackage{dsfont}
\usepackage{graphicx}
\usepackage{subfigure}%jia
\newenvironment{proof}{\paragraph{Proof:}}{\hfill$\square$}
\newtheorem{theorem}{Theorem}[section]
\newtheorem{lemma}{Lemma}[section]

\newtheorem{remark}[theorem]{Remark}

\numberwithin{equation}{section}
\numberwithin{lemma}{section}

\begin{document}
\title{An inverse medium problem using Stekloff eigenvalues and a Bayesian approach}
\author{Juan Liu$^1$, Yanfang Liu$^2$, and Jiguang Sun$^{2}$}
\date{}
\maketitle
\begin{abstract}
This paper studies the reconstruction of Stekloff eigenvalues and the index of refraction of an inhomogeneous medium from Cauchy data. The inverse spectrum problem to reconstruct Stekloff eigenvalues is investigated using a new integral equation for the reciprocity gap method. Given reconstructed eigenvalues, a Bayesian approach is proposed to estimate the index of refraction. Moreover, since it is impossible to know the multiplicities of the reconstructed eigenvalues and since the eigenvalues can be complex, we employ the recently developed spectral indicator method to compute Stekloff eigenvalues. Numerical experiments validate the effectiveness of the proposed methods.
\end{abstract}


\section{Introduction}
Inverse scattering problems for inhomogeneous media have many applications such as medical imaging and nondestructive testing.
In this paper, the inverse spectrum problem to reconstruct the Stekloff eigenvalues from Cauchy data is investigated first 
using a new integral equation for the reciprocity gap method.
Then these eigenvalues are used to estimate the index of refraction of the inhomogeneous medium.
Due to the lack of knowledge of the relation between Stekloff eigenvalues and the index of refraction,
we propose a Bayesian approach.
Since the eigenvalues are complex for absorbing media and the multiplicities are not known,
the recently developed spectral indicator method is employed to compute the Stekloff eigenvalues \cite{Huang2016JCP, Huang2017}.

The reconstruction of certain eigenvalues from the scattering data has been studied by many researchers.
In the context of qualitative methods in inverse scattering, it has been shown that interior eigenvalues such as Dirichlet eigenvalues and
transmission eigenvalues can be determined from the scattering data
\cite{CakoniColtonHaddar2010CRASP, Sun2011IP, LiuSun2014IP} (see also the special issue edited by Lechleiter and Sun \cite{LechleiterSun2017AA}).
A related method, which can be used to compute interior eigenvalues using the scattering data, is the inside-outside duality
\cite{KirschLechleiter2013IP, LechleiterPeters2015CMS, LechleiterRennoch2015, PeterKleefeld2016IP}.

Given reconstructed eigenvalues, a legitimate question is what information about the obstacle can be obtained.
For inhomogeneous non-absorbing media, transmission eigenvalues have been used to reconstruct the shape of the obstacle \cite{Sun2012IP} 
and obtain useful information of the index of refraction
\cite{CakoniColtonMonk2007IP, Sun2011IP, AudibertCakoniHaddar2017, HarrisCakoniSun2014IP, BondarenkoHarrisKleefeld2017AA, LiHuangLinWang2018IPI}.
%In general, a few smallest (real) transmission eigenvalues are reconstructed from the scattering data.
%Then these eigenvalues are used to estimate the lower bound, obtain the constant approximation or test small changes of the index for refraction
%\cite{}.
However, the use of transmission eigenvalues has two drawbacks: 1) multi-frequency data are necessary; and
2) only real transmission eigenvalues can be determined from the scattering data so far. 
% Moreover, the method can only be used for non-absorbing media since transmission eigenvalues do not exist for absorbing media.

It has been shown that Stekloff eigenvalues associated with the scattering problem
can be determined from far field data of a single frequency \cite{CakoniColton2016, AudibertCakoniHaddar2017}.
Unlike transmission eigenvalues, Stekloff eigenvalues exist for absorbing media as well.
Hence the use of Stekloff eigenvalues avoids the above two drawbacks and has the potential to work for a wider class of problems.
In this paper, a new integral equation for the reciprocity gap (RG) method
\cite{ColtonHaddar2005IP, DiCristoSun2006IP, MonkSun2007IPI} is introduced to determine Stekloff eigenvalues from Cauchy data.
Then a Bayesian approach is proposed to estimate the index of refraction.
The Metropolis-Hastings (M-H) Algorithm is used to explore the posterior distribution. Numerical examples show that the proposed methods are effective.
We refer the readers to \cite{JariE.Somersalo2006, Stuart2010AN} and references therein on the Bayesian framework for inverse problems and \cite{BaussardEtal2001IP, GharsalliEtal2014IP, YangMaZheng2015, HarrisRome2017AA} for the Bayesian methods for some inverse scattering problems.

The rest of the paper is organized as follows. In Section 2, the forward scattering problem and the associated Stekloff eigenvalue problem are introduced.
In Section 3, a new integral equation for the reciprocity gap method is proposed to reconstruct Stekloff eigenvalues using Cauchy data.
In Section 4, a Bayesian approach and the MCMC method are proposed to estimate the index of refraction.
Finally, numerical examples are provided in Section 5.

\section{Scattering Problem and Stekloff Eigenvalues}
In this section, we introduce the direct scattering problem, the Stekloff eigenvalue problem,
and the inverse scattering problems using Cauchy data. Then a monotonicity of the largest negative Stekloff eigenvalue is proved.

Let $D$ be a bounded domain in $\mathbb{R}^2$ with boundary $\partial D$ of class $C^2$.
Let $k$ be the wavenumber and $n(x)$ be the index of refraction such that $n(x) \in L^\infty (\mathbb{R}^2)$.
Assume that $n(x)=1$ for $\mathbb{R}^2\setminus \overline{D}$ and $\Re(n(x))>0, \Im(n(x))\geq 0$, where $\Re(\cdot)$ and $\Im(\cdot)$ denote the real and imaginary parts, respectively.
The direct scattering problem is to find the total field $u$ such that
\begin{equation}\label{medium}
  \left\{
   \begin{array}{lll}
   &\Delta u+k^2n(x) u=0,\ \ \ \textrm{in}\  \mathbb{R}^2 \setminus \{x_0\}, \\
   &u=u^s+u^i,\\
   &\lim\limits_{{r}\rightarrow\infty}r^{\frac{1}{2}}(\partial u^s/\partial r-{i} k u^s)=0,\ \ \ r=|x|,
   \end{array}\right.
\end{equation}
where 
\[
u^i:=\Phi(\cdot, x_0) \quad x_0 \in \mathbb{R} \setminus \overline{D}
\]
is the incident wave generated by a point source. Here $\Phi$ is the fundamental solution of the Helmholtz equation.

The associated Stekloff eigenvalue problem is defined as follows \cite{CakoniColton2016}.
% The Stekloff eigenvalue problem is defined as follows \cite{CakoniColton2016}.
Find $\lambda \in \mathbb C$ and a non-trivial function $w$ such that
\begin{equation}\label{stekloff}
  \left\{
   \begin{array}{ll}
   &\Delta w+k^2n(x) w=0,\ \ \ \textrm{in}\  B, \\
   &\partial w/\partial \nu+\lambda w=0,\ \ \ \textrm{on}\ \Gamma,
   \end{array}\right.
\end{equation}
where $B$ be a bounded domain in $\mathbb R^2$ and $\Gamma:=\partial B$ such that $D \subset B$.

\begin{figure}[t]
\centering
\includegraphics[width=0.35\textwidth]{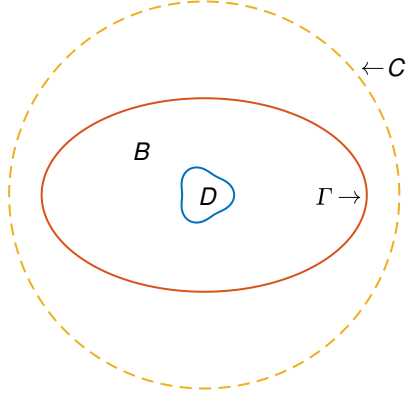}
\caption{\label{fig0}Explicative picture for the scattering problem.}
\end{figure}

Assume that the Cauchy data $u$ and $\partial_\nu u:=\partial u/\partial \nu$ are known on $\Gamma$
for each the incident wave $u^i:=\Phi(\cdot, x_0), x_0\in C$, where $C$ is a simple closed curve containing $B$ (see \ref{fig0}). 
The inverse scattering problems considered in this paper are:
\begin{itemize}
\item[{\bf IP1}] Reconstruct Stekloff eigenvalues from Cauchy data;
\item[{\bf IP2}] Estimate the index of refraction $n(x)$ using Stekloff eigenvalues.
\end{itemize}

The weak formulation for the Stekloff eigenvalue problem \eqref{stekloff} is to find $(\lambda,u)\in \mathbb{C}\times H^1(B)$ such that
\begin{equation}\label{variation}
\left (\nabla w, \nabla v\right)-k^2 \left(nw,v\right)=-\lambda \left \langle w,v\right\rangle\quad \forall \, v\in H^1(B),
\end{equation}
where $\left(f,g\right)=\int_B f\overline{g}d x\ \ \textrm{and}\ \ \left\langle f,g\right\rangle=\int_\Gamma f \overline{g} d s$.

When $n(x)$ is real, all Stekloff eigenvalues are real and they form an infinite discrete set \cite{CakoniColton2016}.
We call $k^2$ a modified Dirichlet eigenvalue of $B$ if there exists a nontrivial $u\in H^1(B)$ such that
\begin{equation}\label{Dirichlet}
  \left\{
   \begin{array}{lll}
   &\Delta u+k^2n u=0,\ \ \ \textrm{in}\  B, \\
   &u=0,\ \ \ \textrm{on}\  \Gamma.
   \end{array}\right.
\end{equation}
\begin{remark}
Note that a standard Dirichlet eigenvalue problem is such that $n(x) \equiv 1$ in \eqref{Dirichlet}. For simplicity, in the rest of the paper, we call $k^2$ in \eqref{Dirichlet}
a Dirichlet eigenvalue.
\end{remark}
It is shown in \cite{AudibertCakoniHaddar2017} that Stekloff eigenvalues accumulate at $-\infty$ if $k^2$ is not a Dirichlet eigenvalue. 
Next, we prove a property of the largest negative Stekloff eigenvalue $\lambda^-_{1}$ when $n(x)$ is given by
\begin{equation}\label{piecewise}
n(x):=n_c=\left\{
\begin{array}{lll}
&1,\ \ \ \textrm{in}\  B\setminus \overline{D}, \\
&c,\ \ \ \textrm{in}\  D.
\end{array}\right.
\end{equation}
Suppose $n_c$ is perturbed by 
\begin{equation*}
\delta n_c:=\left\{
\begin{array}{lll}
&0,\ \ \ \textrm{in}\  B\setminus \overline{D}, \\
&\delta c,\ \ \ \textrm{in}\  D,
\end{array}\right.
\end{equation*}
where $\delta c$ is also a real constant. The perturbation $\delta n_c$ leads to $\delta w$ and $\delta \lambda^-_1$ of the eigenpair.
From \eqref{variation}, $\delta w\in H^1(B)$ and $\delta \lambda^-_1$ satisfies
\begin{equation*}
\big (\nabla (w+\delta w), \nabla v\big)-k^2 \big((n_c+\delta n_c)(w+\delta w),v\big)=-(\lambda^-_{1}+\delta \lambda^-_{1})\big\langle w+\delta w,v\big\rangle\quad\forall \,v\in H^1(B).
\end{equation*}
Using the fact that $(w,\lambda^-_{1})$ is a real eigenpair, we have that
\begin{eqnarray*}
&&\big (\nabla \delta w, \nabla v\big)-k^2 \big(\delta n_c (w+\delta w),v\big)-k^2\big(n_c\delta w,v\big)\\
&&\qquad\qquad\qquad\qquad =-\delta \lambda^-_{1}\big\langle w+\delta w,v\big\rangle-\lambda^-_{1} \big\langle \delta w,v\big\rangle \quad \forall \,v\in H^1(B).
\end{eqnarray*}
Letting $v=w$ and noting that $n_c$ is real, we have that
\begin{equation*}
k^2\big (\delta n_c (w+\delta w), w\big)=\delta \lambda^-_{1}\langle w+\delta w,w\rangle,
\end{equation*}
which implies that
\begin{eqnarray}\label{disturb}
\delta\lambda^-_{1}&=&\frac{k^2\big (\delta n_c (w+\delta w),  w\big)}{\langle w+\delta w,w\rangle}\nonumber\\
&=& \frac{k^2\big (\delta n_c w,  w\big)+k^2\big (\delta n_c \delta w,  w\big)}{\langle w,w\rangle+\langle \delta w,w\rangle}\nonumber\\
&=& \frac{k^2\delta c\big ( w,  w\big)_D+k^2\delta c\big ( \delta w,  w\big)_D}{\langle w,w\rangle+\langle \delta w,w\rangle},
\end{eqnarray}
where $\left(f,g\right)_D=\int_D f\overline{g}d x$. If $\delta c> 0$ is small enough, one has that
\begin{equation}\label{dndww}
|(\delta w,  w)_D|<\frac{1}{2}(w,  w)_D\quad \text{and} \quad |\langle \delta w,w\rangle |<\frac{1}{2}\langle w,w\rangle.
\end{equation}
From \eqref{disturb} and \eqref{dndww}, we have
\begin{equation}\label{relation}
0 < \frac{k^2\delta c\big(w,w\big)_D}{3\langle w,w \rangle} \leq \delta \lambda^-_{1} \leq \frac{3k^2\delta c\big(w,w\big)_D}{\langle w,w \rangle}.
\end{equation}
This implies that $\lambda^-_{1}$ is monotonically increasing with respect to $n_c$. 
This breaks until $k^2$ becomes a (modified) Neumann eigenvalue, i.e., there exists a non-trivial $u$ such that
\begin{equation}\label{Neumann}
  \left\{
   \begin{array}{lll}
   &\Delta u+k^2n_c u=0,\ \ \ \textrm{in}\  B, \\
   &\frac{\partial u}{\partial \nu}=0,\ \ \ \textrm{on}\  \Gamma.
   \end{array}\right.
\end{equation}
Note that a standard eigenvalues is $k^2$ satisfying \eqref{Neumann} for $n_c \equiv 1$. Again, we $k^2$ a Neumann eigenvalue for simplicity.
Excluding this case, we actually proved the following theorem.

\begin{theorem}\label{ThmNS}
Let the index of refraction be defined in \eqref{piecewise} and $[a,b]$ be an interval that $k^2$ is not a Neumann eigenvalue of \eqref{Neumann} for any $c\in [a,b]$. 
Then the largest negative Stekloff eigenvalue $\lambda^-_{1}$ is monotonically increasing on $[a,b]$.
\end{theorem}
Assume that the largest negative Stekloff eigenvalue $\lambda^-_n$ is obtained.
If the shape of $D$ is known,
$\lambda^-_1$ uniquely determines $n_c$ on some suitable interval $[a, b]$ by \ref{ThmNS}.
However, it is not true on $\mathbb R$ as it is known that different $n_c$'s can give the same $\lambda^-_1$.
%%%%%%%%%%%%%%%%%%%%%%%%%%%%%%%%%%%%%%%%%%%%%%%%%%%%%%%%%%%%%%

\section{Reconstruction of Stekloff Eigenvalues}\label{RGA}
Now we consider {\bf IP1} to reconstruct Stekloff eigenvalues from Cauchy data. 
The main ingredient is the reciprocity gap method using Cauchy data \cite{ColtonHaddar2005IP, DiCristoSun2006IP, MonkSun2007IPI}.
Assume that $u$ and $\partial_\nu u:=\partial u/\partial \nu$ are known on $\Gamma$
for each point source incident wave $u^i:=\Phi(\cdot, x_0), x_0\in C$ (see \ref{fig0}).
The following auxiliary scattering problem will be useful in the subsequent analysis.
Find $
u_\lambda(\cdot,x_0):=u_\lambda^s(\cdot,x_0)+\Phi(\cdot,x_0)
$
such that
\begin{equation}\label{auxiliary}
  \left\{
   \begin{array}{lll}
   &\Delta u_\lambda+k^2 u_\lambda=0,\ \ \ \textrm{in}\  \mathbb{R}^2\setminus \{\overline{B} \cup \{x_0\} \}, \\
   &\partial_\nu u_\lambda+\lambda u_\lambda=0, \ \ \ \textrm{on}\  \Gamma,\\
   &\lim\limits_{{r}\rightarrow\infty}r^{\frac{1}{2}}(\partial u_\lambda^s/\partial r-{i} k u_\lambda^s)=0,\ \ \ r=|x|,
   \end{array}\right.
\end{equation}
where $\nu$ is the unit outward normal to $\Gamma$ and $\lambda$ is a constant such that $\Im(\lambda)\geq 0$.
It is shown in \cite{CakoniColton2016} that \eqref{auxiliary} has a unique solution.

Denote by $U$ and $U_\lambda$ the sets of solutions $u(x, x_0)$ to \eqref{medium} and $u_\lambda(x,x_0)$ to \eqref{auxiliary}
for the incident wave $\Phi(\cdot, x_0), x_0\in C$, respectively.
Define the reciprocity gap functional by
\begin{equation}
R(v_1,v_2)=\int_\Gamma (v_1\partial_\nu v_2-v_2\partial_\nu v_1)ds,
\end{equation}
where $v_1$ and $v_2$ are solutions of the Helmholtz equation.
Let $\mathbb{S}:=\{d\in \mathbb{R}^2;|d|=1\}$ and consider the integral equation of finding $g\in L^2(\mathbb{S})$ to
\begin{equation}\label{IE}
R(u_\lambda(\cdot,x_0)-u(\cdot,x_0),v_g(\cdot))=R(u_\lambda(\cdot,x_0),\Phi_z(\cdot)) \quad \forall x_0\in C,
\end{equation}
where $\Phi_z(\cdot):=\Phi(\cdot,z)$ for some $z\in B$ and $v_g$ is the Herglotz wave function defined by
\[
v_g(x):=\int_{\mathbb{S}}{e}^{{i}kx\cdot d}g(d)ds(d).
\]

\begin{lemma}\label{lemma1}
If $\int_\Gamma u_\lambda(x,x_0)f(x)ds(x)=0$ for all $u_\lambda\in U_\lambda$, then $f(x)=0$ on $\Gamma$.
\end{lemma}
\begin{proof}
Assume that $f(x)$ satisfies $\int_\Gamma u_\lambda(x,x_0)f(x)ds(x)=0$ for all $x_0\in C$. Let $\tilde{u}^s$ be the solution of the following problem
\begin{equation}
  \left\{
   \begin{array}{lll}
   &\Delta \tilde{u}^s+k^2\tilde{u}^s=0,\ \ \ \textrm{in}\  \mathbb{R}^2\setminus \overline{B}, \\
   &\partial_\nu \tilde{u}^s+\lambda \tilde{u}^s=f,\ \ \ \textrm{on } \Gamma,\\
   &\lim\limits_{{r}\rightarrow\infty}r^{\frac{1}{2}}(\partial \tilde{u}^s/\partial r-{i} k \tilde{u}^s)=0,\ \ \ r=|x|.
   \end{array}\right.
\end{equation}
Using Green's representation theorem \cite{ColtonKress2013}, Green's second theorem 
and the boundary condition $\partial_\nu u_\lambda+\lambda u_\lambda=0$ on $\Gamma$ for all $x_0\in C$, we have that
\begin{eqnarray}
\tilde{u}^s(x_0)&=&\int_\Gamma \partial_\nu\Phi(x,x_0) \tilde{u}^s(x)-\Phi(x,x_0)\partial_\nu \tilde{u}^s(x) d s\nonumber\\
&=&\int_\Gamma \partial_\nu\Phi(x,x_0) \tilde{u}^s(x)-\Phi(x,x_0)\partial_\nu \tilde{u}^s(x) d s\nonumber\\
&&\qquad +\int_\Gamma\partial_\nu u_\lambda^s(x) \tilde{u}^s(x)-u_\lambda^s(x)\partial_\nu \tilde{u}^s(x) d s\nonumber\\
&=&-\int_\Gamma u_\lambda(x,x_0)\partial_\nu \tilde{u}^s(x)-\tilde{u}^s(x) \partial_\nu u_\lambda(x,x_0) d s \nonumber\\
&=&-\int_\Gamma u_\lambda(x,x_0)(\partial_\nu \tilde{u}^s(x)+\lambda \tilde{u}^s(x)) d s \nonumber\\
&=&-\int_\Gamma u_\lambda(x,x_0) f(x) d s \nonumber\\
&=&0.
\end{eqnarray}
The unique continuity principle implies that $\tilde{u}^s(x)=0$ in $\mathbb{R}^2\setminus \bar{B}$. By the trace theorem we have $f=0$ on $\Gamma$.
\end{proof}

\begin{theorem}\label{theorem1}
If $\lambda$ is not a Stekloff eigenvalue of \eqref{stekloff}, then for $u\in U$ and $u_\lambda\in U_\lambda$, 
the operator $\mathcal{R}: L^2(\mathbb{S})\rightarrow L^2(C)$ defined by
\begin{equation*}
\mathcal{R}(g):=R\big(u_\lambda(\cdot,x_0)-u(\cdot,x_0),v_g(\cdot)\big), \quad x_0\in C
\end{equation*}
is injective.
\end{theorem}
\begin{proof}
Let $g$ satisfy $R(u_\lambda(\cdot,x_0)-u(\cdot,x_0),v_g(\cdot))=0$ for all $x_0\in C$. If $g\neq 0$, let $w^s$ solve
\begin{equation}\label{eq2.7}
  \left\{
   \begin{array}{ll}
   &\Delta w^s+k^2 n w^s =k^2 (1-n)v_g,\ \ \ \textrm{in}\  \mathbb{R}^2, \\
   &\lim\limits_{{r}\rightarrow\infty}r^{\frac{1}{2}}(\partial w^s/\partial \nu-{i} kw^s)=0,\ \ \ r=|x|.
   \end{array}\right.
\end{equation}
Using the boundary condition $\partial_\nu u_\lambda+\lambda u_\lambda=0$ on $\Gamma$ and Green's second theorem twice, the following holds
\begin{eqnarray}\label{eq2.8}
&&\int_\Gamma u_\lambda \big(\partial_\nu (w^s+v_g)+\lambda (w^s+v_g)\big) ds\nonumber\\
&=&\int_\Gamma \big[u_\lambda\partial_\nu(w^s+v_g)-(w^s+v_g)\partial_\nu u_\lambda\big]d s\nonumber\\
&=&\int_\Gamma \big[u_\lambda\partial_\nu(w^s+v_g)-(w^s+v_g)\partial_\nu u_\lambda\big] d s-\int_\Gamma \big[u\partial_\nu(w^s+v_g)-(w^s+v_g)\partial_\nu u \big]d s\nonumber\\
&=&\int_\Gamma\big[ (u_\lambda-u)\partial_\nu(w^s+v_g)-(w^s+v_g)\partial_\nu (u_\lambda-u) \big]d s\nonumber\\
&=&\int_\Gamma\big[ (u_\lambda-u)\partial_\nu v_g-v_g\partial_\nu (u_\lambda-u) \big]d s
+\int_\Gamma\big[ (u_\lambda-u)\partial_\nu w^s-w^s\partial_\nu (u_\lambda-u) \big]d s\nonumber\\
&=&\int_\Gamma\big[ (u_\lambda-u)\partial_\nu v_g-v_g\partial_\nu (u_\lambda-u) \big]d s\nonumber\\
&=&R(u_\lambda-u,v_g)=0.
\end{eqnarray}
From \eqref{eq2.8} and \ref{lemma1}, we have $\partial_\nu (w^s+v_g)+\lambda (w^s+v_g)=0$ on $\Gamma$. 
Together with \eqref{eq2.7}, $w^s+v_g$ satisfies
\begin{equation}\label{eq2.9}
  \left\{
   \begin{array}{ll}
   &\Delta (w^s+v_g)+k^2 n (w^s+v_g) =0,\ \ \ \textrm{in}\  B, \\
   &\partial_\nu (w^s+v_g)+\lambda (w^s+v_g)=0\ \ \ \textrm{on}\  \Gamma.
   \end{array}\right.
\end{equation}
Since $\lambda$ is not a Stekloff eigenvalue, \eqref{eq2.9} only has the trivial solution $w^s+v_g=0$ in $B$.
From \eqref{eq2.7} and the unique continuity principle, $w^s+v_g=0$ in $\mathbb{R}^2$, i.e., the Herglotz wave function $v_g=-w^s$ satisfies the radiation condition. 
This is a contradiction.
\end{proof}

The following theorem is the main result on the reconstruction of Stekloff eigenvalues from Cauchy data.
\begin{theorem}\label{theorem2}
\begin{itemize}
\item[1.] If $\lambda$ is not a Stekloff eigenvalue of \eqref{stekloff} and $z\in B$, then there exists a sequence $\{g_n\}, g_n\in L^2(\mathbb{S})$, such that
\begin{equation}\label{eq2.10}
\lim\limits_{n\rightarrow \infty}R\big(u_\lambda-u,v_{g_n}\big)=R(u_\lambda,\Phi_z),\ \ \ \ u_\lambda\in U_\lambda,\  u\in U
\end{equation}
and $v_{g_n}$ converges in $L^2(B)$.
\item[2.] If $\lambda$ is a Stekloff eigenvalue, then for every sequence $\{g_n^z\}, g_n^z\in L^2(\mathbb{S})$ satisfying
\begin{equation}\label{eq2.11}
\lim\limits_{n\rightarrow \infty}R\big(u_\lambda-u,v_{g_n^z}\big)=R(u_\lambda,\Phi_z), \ \ \ \ u_\lambda\in U_\lambda,\  u\in U,
\end{equation}
$\lim\limits_{n\rightarrow \infty}\|v_{g_n^z}\|_{H^1(B)}=\infty$ for almost every $z\in B$.
\end{itemize}
\end{theorem}

\begin{proof}
1. Let $w_z$ be the solution of the following problem
\begin{equation}\label{eq2.12}
  \left\{
   \begin{array}{ll}
   &\Delta w_z+k^2 n w_z =0,\ \ \ \textrm{in}\  B, \\
   &\partial_\nu w_z+\lambda w_z=\partial_\nu \Phi(\cdot,z)+\lambda \Phi(\cdot,z)\ \ \ \textrm{on}\  \Gamma.
   \end{array}\right.
\end{equation}
From Lemma 3.1 of \cite{CakoniColton2016}, we have
\begin{equation*}
w_z=w_z^i+w_z^s,
\end{equation*}
where $w_z^i$ satisfies the Helmholtz equation in $B$ and $w_z^s\in H_{loc}^2(\mathbb{R}^2)$ is a radiation solution to
\begin{equation}\label{eq2.13}
  \left\{
   \begin{array}{ll}
   &\Delta w^s_z+k^2 n w^s_z =k^2 (1-n)w_z^i,\ \ \ \textrm{in}\  \mathbb{R}^2, \\
   &\lim\limits_{{r}\rightarrow\infty}r^{\frac{1}{2}}(\partial w^s_z/\partial r-{i} kw^s_z)=0,\ \ \ r=|x|.
   \end{array}\right.
\end{equation}
Due to the denseness property  (Theorem 5.21 of \cite{ColtonKress2013}), there exists a sequence of Herglotz wave functions $\{v_{g_n}\}$ such that
\begin{equation}\label{eq2.13b}
v_{g_n}+w_z^s\rightarrow w_z^i+w_z^s=w_z,\ \ \ \ n\rightarrow \infty.
\end{equation}
Next we show that $\{v_{g_n}\}$ satisfies $\lim\limits_{n\rightarrow \infty}R\big(u_\lambda-u,v_{g_n}\big)=R(u_\lambda,\Phi_z)$.
Using Green's second theorem twice, one has that
\begin{eqnarray}\label{eq2.14}
&&\lim\limits_{n\rightarrow \infty}R(u_\lambda-u,v_{g_n})-R(u_\lambda,\Phi_z)\nonumber\\
&=&\lim\limits_{n\rightarrow \infty}\int_{\Gamma}\big[(u_\lambda-u)\partial_\nu v_{g_n}-v_{g_n}\partial_\nu(u_\lambda-u) \big]d s-\int_\Gamma \big[u_\lambda\partial_\nu \Phi_z-\Phi_z\partial_\nu u_\lambda \big]d s\nonumber\\
&=&\lim\limits_{n\rightarrow \infty}\int_{\Gamma}\big[(u_\lambda-u)\partial_\nu (v_{g_n}+w_z^s)-(v_{g_n}+w_z^s)\partial_\nu(u_\lambda-u) \big]d s\nonumber\\
&&\qquad -\int_\Gamma \big[u_\lambda\partial_\nu \Phi_z-\Phi_z\partial_\nu u_\lambda \big]d s\nonumber\\
&=&\lim\limits_{n\rightarrow \infty}\int_{\Gamma}\big[u_\lambda\partial_\nu (v_{g_n}+w_z^s)-(v_{g_n}+w_z^s)\partial_\nu u_\lambda \big]d s\nonumber\\
&&\qquad -\lim\limits_{n\rightarrow \infty}\int_{\Gamma}\big[u\partial_\nu (v_{g_n}+w_z^s)-(v_{g_n}+w_z^s)\partial_\nu u \big]d s
 -\int_\Gamma\big[ u_\lambda\partial_\nu \Phi_z-\Phi_z\partial_\nu u_\lambda \big]d s\nonumber\\
&=&\lim\limits_{n\rightarrow \infty}\int_{\Gamma}\big[u_\lambda\partial_\nu (v_{g_n}+w_z^s)-(v_{g_n}+w_z^s)\partial_\nu u_\lambda \big]d s-\int_\Gamma\big[ u_\lambda\partial_\nu \Phi_z-\Phi_z\partial_\nu u_\lambda \big]d s\nonumber\\
&=&\lim\limits_{n\rightarrow \infty}\int_{\Gamma}\big[u_\lambda\partial_\nu (v_{g_n}+w_z^s-\Phi_z)-(v_{g_n}+w_z^s-\Phi_z)\partial_\nu u_\lambda \big]d s\nonumber\\
&=&\lim\limits_{n\rightarrow \infty}\int_{\Gamma}u_\lambda\big[\partial_\nu (v_{g_n}+w_z^s-\Phi_z)+\lambda(v_{g_n}+w_z^s-\Phi_z)\big]d s\nonumber\\
&=&0,
\end{eqnarray}
where the last step is due to \eqref{eq2.12} and \eqref{eq2.13b}.

2. Assume on the contrary that for $z\in B_\rho$, where $B_\rho\subset B$ is a small ball of radius $\rho$,
$\|v_{g_n^z}\|_{H^1(B)}$ is bounded as $n\rightarrow \infty$. Then there exists a subsequence of $v_{g_n^z}$, still denoted by $v_{g_n^z}$,
converging weakly to a function $v^i\in H^1(B)$. Then
\begin{equation}\label{eq2.15_1}
\int_\Gamma \big[(u_\lambda-u)\partial_\nu v^i-v^i\partial_\nu(u_\lambda-u)\big]ds-\int_\Gamma \big[u_\lambda\partial \Phi_z-\Phi_z\partial_\nu u_\lambda\big]ds=0.
\end{equation}
Let $w^s\in H_{loc}^2(\mathbb{R}^2)$ be a radiating solution to
\begin{equation}\label{eq2.15}
  \left\{
   \begin{array}{ll}
   &\Delta w^s+k^2 n w^s =k^2 (1-n)v^i,\ \ \ \textrm{in}\  \mathbb{R}^2, \\
   &\lim\limits_{{r}\rightarrow\infty}r^{\frac{1}{2}}(\partial w^s/\partial r-{i} kw^s)=0,\ \ \ r=|x|.
   \end{array}\right.
\end{equation}
From the Green's second theorem and \eqref{eq2.15_1}, $w:=v^i+w^s$ satisfies
\begin{eqnarray}\label{eq2.16}
&&\int_\Gamma u_\lambda\big[\partial_\nu (w-\Phi_z)+\lambda(w-\Phi_z)\big]ds\nonumber\\
&=&\int_\Gamma\big[u_\lambda\partial_\nu w-w\partial_\nu u_\lambda\big]ds-\int_\Gamma\big[u_\lambda\partial_\nu \Phi_z-\Phi_z \partial_\nu u_\lambda \big]ds\nonumber\\
&=&\int_\Gamma\big[(u_\lambda-u)\partial_\nu w-w\partial_\nu (u_\lambda-u)\big]ds-\int_\Gamma\big[u_\lambda\partial_\nu \Phi_z-\Phi_z \partial_\nu u_\lambda \big]ds\nonumber\\
&=&\int_\Gamma\big[(u_\lambda-u)\partial_\nu w^s-w^s\partial_\nu (u_\lambda-u)\big]ds\nonumber\\
&&+\int_\Gamma\big[(u_\lambda-u)\partial_\nu v^i-v^i\partial_\nu (u_\lambda-u)\big]ds-\int_\Gamma\big[u_\lambda\partial_\nu \Phi_z-\Phi_z \partial_\nu u_\lambda\big]ds\nonumber\\
&=&\int_\Gamma\big[(u_\lambda-u)\partial_\nu w^s-w^s\partial_\nu (u_\lambda-u)\big]ds\nonumber\\
&=&0.
\end{eqnarray}
From \eqref{eq2.16} and \ref{lemma1},
\begin{equation*}
\partial_\nu (w-\Phi_z)+\lambda(w-\Phi_z)=0\ \ \ \ \textrm{on}\ \Gamma,
\end{equation*}
which, together with \eqref{eq2.15}, implies that $w$ satisfies
\begin{equation}\label{eq2.17}
  \left\{
   \begin{array}{ll}
   &\Delta w+k^2 n w =0\ \ \ \textrm{in}\  B, \\
   &\partial_\nu w+\lambda w=\partial_\nu \Phi(\cdot,z)+\lambda \Phi(\cdot,z)\ \ \ \textrm{on}\  \Gamma.
   \end{array}\right.
\end{equation}
From \ref{lemma1} and the proof of Theorem 3.3 of \cite{CakoniColton2016}, \eqref{eq2.17} is solvable if and only if
\begin{equation}\label{eq2.18}
  \int_\Gamma\bigg(\frac{\partial \Phi(\cdot,z)}{\partial \nu}+\lambda\Phi(\cdot,z)\bigg){\omega}_\lambda ds=0
\end{equation}
for each Stekloff eigenfunction ${\omega}_\lambda\in H^1(B)$.
Since ${\omega}_\lambda$ satisfies $\partial_\nu {\omega}_\lambda+\lambda {\omega}_\lambda=0$ on $\Gamma$, \eqref{eq2.18} becomes
\begin{equation*}
\int_\Gamma \bigg(\frac{\partial \Phi(\cdot,z)}{\partial \nu} {\omega}_\lambda-\Phi(\cdot,z)\frac{\partial {\omega}_\lambda}{\partial \nu}\bigg)ds=0.
\end{equation*}
Green's representation theorem implies that $\omega_\lambda(z)=0$ for $z\in B_\rho$. The unique continuation principle now implies that the Stekloff eigenfunction $\omega_\lambda=0$ in $B$, which is a contradiction.
\end{proof}

Based on \ref{theorem2}, the following reciprocity gap algorithm can be used to
reconstruct (several) Stekloff eigenvalues from Cauchy data.

\vskip 0.2cm
{\bf The RG Algorithm}\label{rgmethod}
\begin{itemize}
\item[1.] For a region of interests (e.g., an interval on $\mathbb R$ for real Stekloff eigenvalues or a rectangular region on $\mathbb C$ for complex Stekloff eigenvalues), generate a grid $T$.

\item[2.] For each $\lambda\in T$, solve the scattering problem \eqref{auxiliary} to compute the auxiliary Cauchy data $u_\lambda(\cdot,x_0)$ and $\partial_\nu u_\lambda(\cdot,x_0)$ on $\Gamma$.

\item[3.] Fix a point $z\in B$, use the Tikhonov regularization to compute
an approximate solution $g_\lambda\in L^2(\mathbb{S})$ to the integral equation
\begin{equation}\label{RGequation}
R\big(u_\lambda(\cdot,x_0)-u(\cdot,x_0),v_{g_\lambda}(\cdot)\big)=R\big(u_\lambda(\cdot,x_0),\Phi(\cdot,z)\big)\quad \forall x_0\in C,
\end{equation}
where $u(\cdot, x_0)$ is the solution to \eqref{medium} for $\Phi(\cdot, x_0), x_0\in C$.

\item[4.] Choose $\lambda$ as a Stekloff eigenvalue of \eqref{stekloff} if the norm of $g_\lambda$ is significantly larger (see Section 5.1).
\end{itemize}

\begin{remark}
The constructed solutions to the reciprocity gap equation may not form a divergent Herglotz wave function series.
Hence the above numerical algorithm might not be able to construct all the eigenvalues.
\end{remark}

\section{Reconstruction of the Index of Refraction}\label{bayes_for}
The algorithm in the previous section can reconstruct Stekloff eigenvalues using Cauchy data.
Given these reconstructed eigenvalues, in this section, we turn to {\bf IP2} to estimate the index of refraction.
The relation between the index of refraction and Stekloff eigenvalues is complicated and, to a large extend, unknown.
Even when $n(x)$ is constant, a single Stekloff eigenvalue cannot uniquely determine it.
Note that \ref{ThmNS} only holds on an appropriate interval.

To this end, we resort to the Bayesian approach, which has been popular for solving inverse problems in recent years \cite{JariE.Somersalo2006, Stuart2010AN}.
% This approach allows us to explore the characterization of all possible solutions and their relative probability distributions.
Firstly, the inverse problem is reformulated as a statistical inference for the index of refraction
using a Bayes formula. Then the Metropolis-Hastings algorithm is employed to explore the posterior distribution of $n(x)$.

\subsection{Bayesian Formulation}
Denote by $\mathcal{N}$ the normal distribution and $\mathcal{U}$ the uniform distribution. {\bf IP2} can be written as
the statistical inference of $n(x)$ such that
~\begin{equation} \label{eq 3.21}
{\boldsymbol \lambda}=\mathcal{G}(n)+E,
\end{equation}
where ${\boldsymbol \lambda} \in \mathbb{C}^m$ is a vector of (reconstructed) Stekloff eigenvalues, $n(x)$ is a random function, $\mathcal{G}$ is the
operator mapping $n(x)$ to ${\boldsymbol \lambda}$ based on \eqref{stekloff}, and $E$ is the random noise.
The noise $E \sim \mathcal{N}(0,\,\sigma^{2})$, which is modeled as additive and mutually independent of $n(x)$.
In the Bayesian framework, the prior information can be coded into the prior density $\pi _{pr} (n)$.
For example, if $n$ is known to be a real constant $n_0$ such that $a < n_0 < b$,
one may take the prior as the continuous uniform distribution, i.e.,  $n \sim \mathcal{U}(a,b)$.

Given Stekloff eigenvalues ${\boldsymbol \lambda}$, the goal of the Bayesian inverse problem is to seek statistical information of $n(x)$ 
by exploring the conditional probability distribution
 $\pi_{post}(n|{\boldsymbol \lambda})$, called the posterior distribution of $n$.
 An important quantity is the conditional mean (CM) of $n$ defined as
\begin{equation} \label{nCM} %eq 3.24
n_{CM}={\mathbb E}\{n|{\boldsymbol \lambda}\}= \int _{\mathbb{R}}n \pi_{post} (n|{\boldsymbol \lambda}) dn,
\end{equation}
which is an constant estimation of $n(x)$. If  $n \sim \mathcal{U}(a,b)$,
by the Bayes formula, the posterior distribution satisfies
\begin{equation} \label{eq 3.22}
\pi_{post}(n|{\boldsymbol \lambda})\propto \mathcal{N}({\boldsymbol \lambda}-\mathcal{G}(n),\,\sigma^{2}) \times \mathcal{U}(a,b),
\end{equation}
i.e.,
\begin{equation} \label{pipostnlambda}
\pi_{post}(n|{\boldsymbol \lambda})\propto \exp \Big(-\frac{1}{2 \sigma^{2} }|{\boldsymbol \lambda}-\mathcal{G}(n)|\Big)\times I(a\leq n\leq b),
\end{equation}
where $I$ is the density function for $\mathcal{U}(a,b)$.
% The inverse problem is to   the posterior distribution $\pi(n|{\boldsymbol \lambda})$.

% An advantage of this estimate is that the smoothness properties of the posterior distribution are not strong.
% CM estimate ends up with performing an integration which we can not compute it directly, because  this posterior distribution $\pi(n|\lambda)$ is not exact. In order to derive the expectation of $n$,

\subsection{Markov Chain Monte Carlo Method}
To explore $\pi_{post}(n|{\boldsymbol \lambda})$ given in \eqref{pipostnlambda}, we employ the popular MCMC (Markov Chain Monte Carlo).
MCMC to estimate CM is as follows: design a Markov Chain $\{ X_{j}\}_{j=0}^{\infty}$ from required distribution
and take the mean of the chain to approximate expectation.
In particular, one could estimate ${\mathbb E}\{n|{\boldsymbol \lambda}\}$ by a sample mean using Monte Carlo integration:
\begin{equation} \label{eq 5}
{\mathbb E}\{n|{\boldsymbol \lambda}\} \approx \frac{1}{m}\sum_{j=1}^{m}n_{j},
\end{equation}
where $n_{j}$, $j=1,\cdots,m$, are samples drawing from $\pi_{post}(n|{\boldsymbol \lambda})$.
Two popular methods are Metropolis-Hasting (M-H) algorithm \cite{Metropolis1953,Hastings1970}
and Gibbs sampler \cite{GemanGeman1987}.
In this paper, we choose a delayed rejection adaptive Metropolis-Hasting algorithm \cite{HarrioEtal2006SC}.

\vskip 0.2cm
{\bf The MH Algorithm}\label{MHmethod}
\begin{enumerate}
\item Choose initial value $n_1 \in \mathbb R$ and set $j=1$;
\item Draw a sample $w$ from a proposal distribution
	\[
		q(n_j, w) \propto \text{exp}\left(-\frac{1}{2\gamma^2}|n_j - w|^2 \right),
	\]
	and compute
	\[
		\alpha(n_j, w) = \min \left(1, \frac{\pi_{post}(w| {\boldsymbol \lambda})}{\pi_{post}(n_j|{\boldsymbol \lambda})} \right);
	\]
\item Draw $t \sim \mathcal{U}(0, 1)$;
\item If $\alpha(n_j, w) \ge t$, set $n_j = w$, else $n_{j+1}=n_j$.
\item When $j=K$, the maximum sample size, stop; else, $j \leftarrow j+1$ and go to 2.
\end{enumerate}

\subsection{Spectral Indicator Method}
In the above algorithm, for each sample $n_j$, one needs to compute Stekloff eigenvalues, which are done by the
finite element method for \eqref{stekloff} \cite{SunZhou2016, LiuSunTurner2018}.
Note that the Stekloff eigenvalues are complex if $n(x)$ is a complex function. 
In addition, the reconstructed eigenvalues only approximate the exact ones and the multiplicities are not known in general.
Thus the numerical methods need to compute Stekloff eigenvalues of \eqref{stekloff} in a region on $\mathbb C$ close to the origin. 

A recently developed spectral indicator method (SIM) is a good fit for this case \cite{Huang2016JCP, Huang2017}.
Given reconstructed eigenvalues ${\boldsymbol \lambda}$, a rectangular region containing these eigenvalues is chosen.
Then SIM is used to compute all eigenvalues inside the region.

\section{Numerical Examples}
In this section, we present some numerical examples to use the RG method to reconstruct Stekloff eigenvalues from the Cauchy data
and estimate the index of refraction using the Bayesian approach.
Three scatterers are considered: a disc with radius $1$ centered at the origin, a square with vertices given by
\begin{equation}\label{square}
(0;-1),\ \ (1;0),\ \ (0;1),\ \ (-1;0),
\end{equation}
and an L-shaped domain given by
\begin{equation}\label{L}
(-0.9,1.1)\times (-1.1,0.9)\setminus [0.1,1,1]\times [-1.1,-0.1].
\end{equation}
Three different indices of refraction are chosen: i) $n(x)=5$, ii) $n(x)=4+2|x|$, and iii) $n(x)=2+4i$.

The synthetic scattering data is simulated by a finite element method with a perfectly matched layer (PML) for \eqref{medium} \cite{chenliu2005}.
The wavenumber is $k=1$. There are $100$ source points uniformly distributed on the curve $C$, a circle with radius $3$.
We compute the Cauchy data at $100$ points uniformly distributed on $\Gamma:=\partial B$ (a circle with radius $2$) and add $3\%$ noise.

\subsection{Reconstruction of Stekloff eigenvalues}\label{part1_exp}
Given Cauchy data, we show that several Stekloff eigenvalues close to the origin can be reconstructed effectively using the RG Algorithm presented 
in Section~\ref{RGA}.
For real $n(x)$, all the eigenvalues are real. We choose an interval $[-5, 5]$ and use a uniform grid $T$ given by
\[
T:=\{\lambda_{m}=-5+0.02m, \quad m=0, 1, \cdots, 500\}.
\]
For complex $n(x)$, the Skeloff eigenvalues are complex. We choose a domain $[-1, 0.5] \times [-0.5, 1]$ and
\[
T:=\{\lambda_{m_1,m_2}=(-1+0.02m_1)+i(-0.5+0.02m_2), \quad m_1, m_2=0, 1, \cdots, 75\}.
\]
Using a fixed point $z=(0.2, 0.6)$ in $B$, for each $\lambda\in T$,
a discretization of \eqref{RGequation} leads to a linear system
\begin{equation}\label{Alambdagflambda}
A^\lambda g_{\lambda}=f^\lambda,
\end{equation}
where $A^\lambda$ is a matrix given by
\[
A^\lambda_{l,j}=R\big(u_\lambda(x,x_l)-u(x,x_l),\exp({i}kx\cdot {d}_j)\big),\quad  l, j=1,2,\cdots,100,
\]
and $f^\lambda$ is a vector given by
\[
f^\lambda_l=R\big(u_\lambda(x,x_l),\Phi(x,z)\big),\quad l=1,2,\cdots,100.
\]
The Tikhonov regularization with the parameter $\alpha=10^{-5}$ is used
to compute an approximate solution $g_\lambda$ to \eqref{Alambdagflambda}:
\begin{equation*}
g_\lambda\approx \big((A^\lambda)^*A^\lambda+\alpha A^\lambda\big)^{-1}(A^\lambda)^*f^\lambda.
\end{equation*}

In the following examples, we show the plots of $|g_\lambda|$.
In all the figures, the crosses are the exact Stekloff eigenvalues computed using a finite element method \cite{SunZhou2016, LiuSunTurner2018}.
The results indicate that Stekloff eigenvalues close to the origin can be determined effectively.

{\bf Example 1} Real index of refraction $n(x)=5$.
The exact and reconstructed eigenvalues are shown in \ref{n5}. The plots of $|g_\lambda|$ for three domains are shown in \ref{fn5}.
\begin{table}[h!]
\centering
\begin{tabular}{l|l|l|l|l|l}\hline
disc & $1.2937$ & $-0.4763$& $-0.5839$ & $-1.2301$ &\\
& ($1.30$) & ($-0.48$) & ($-0.58$) & ($-1.23$)&\\\hline
square & $0.3792$ & $-0.5418$ & $-0.6148$& & \\
&($0.38$)& ($-0.54$)& ($-0.62)$& &\\\hline
L-shaped & $3.0218$& $0.6882$& $-0.5266$& $-0.5834$& $-1.2188$ \\
&($3.04$)& ($0.70$)& ($-0.52$)&($-0.58$)& ($-1.20$) \\\hline
\end{tabular}
\caption{\label{n5} The exact Stekloff eigenvalues and their reconstructions (in the parentheses)  for $n(x)=5$.}
\end{table}

\begin{figure}[h!]
  \centering
  \subfigure[\textbf{disc}]{
    \includegraphics[width=2.5in]{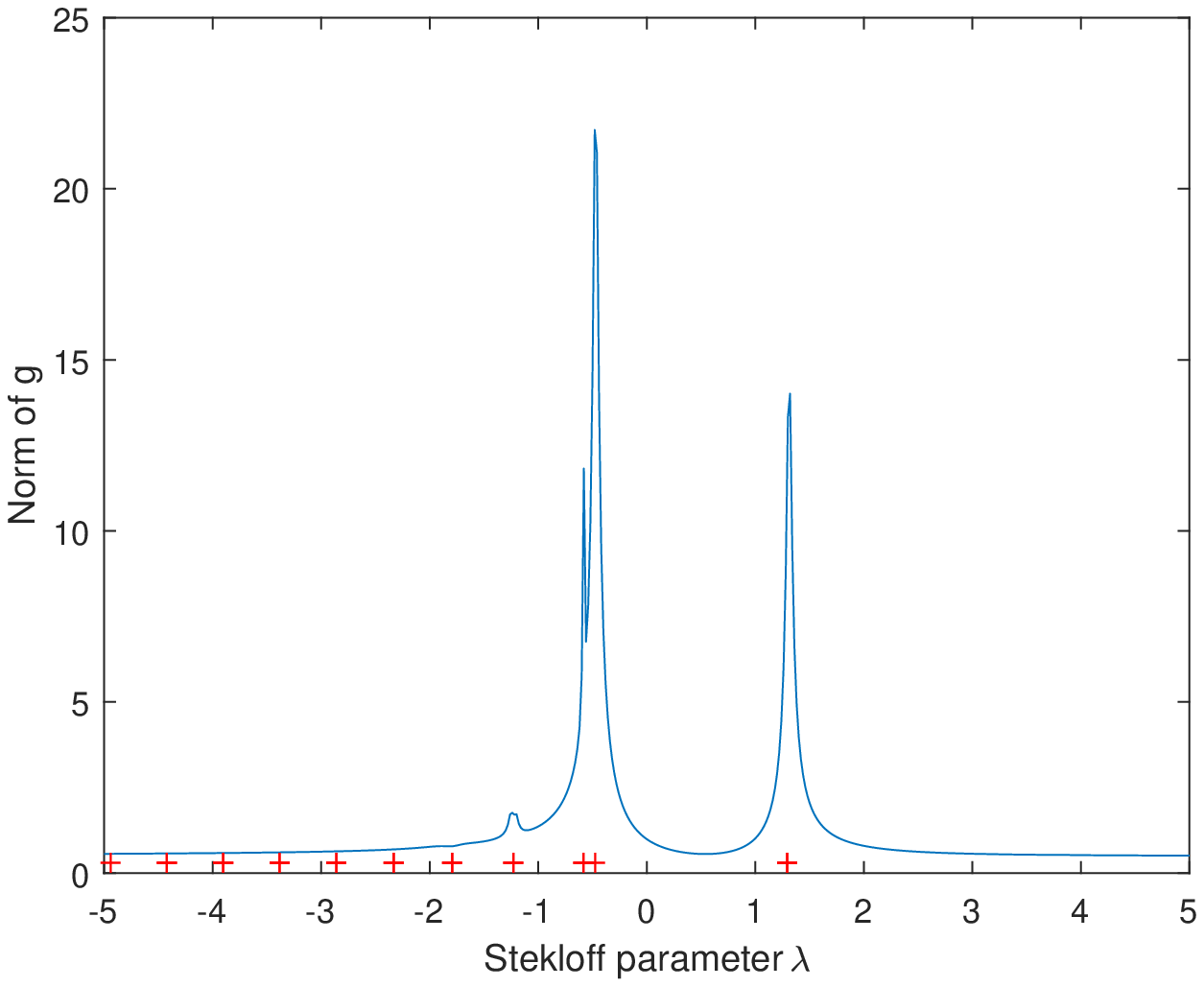}}
  \subfigure[\textbf{square}]{
    \includegraphics[width=2.5in]{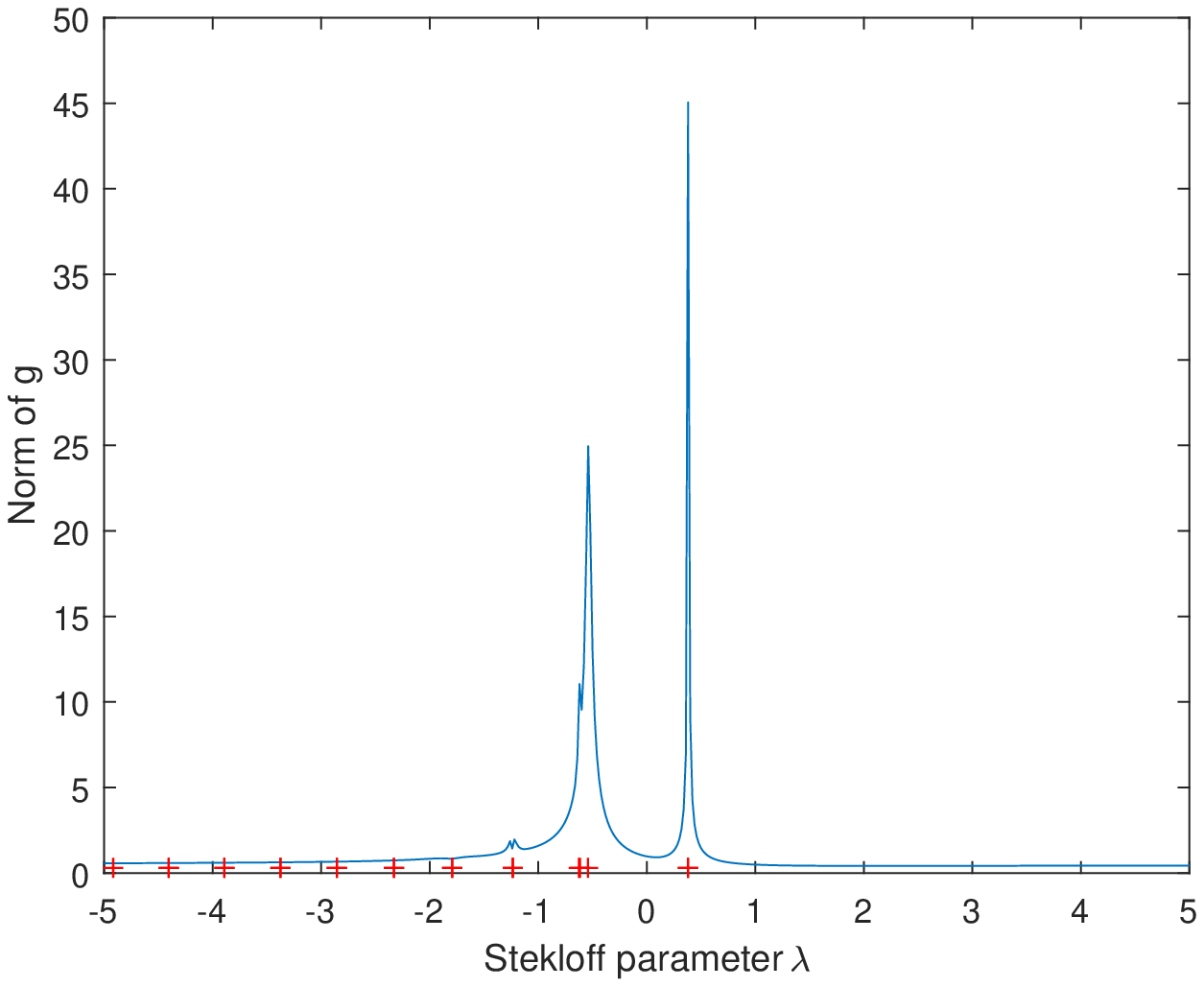}}
  \subfigure[\textbf{L-shaped}]{
    \includegraphics[width=2.5in]{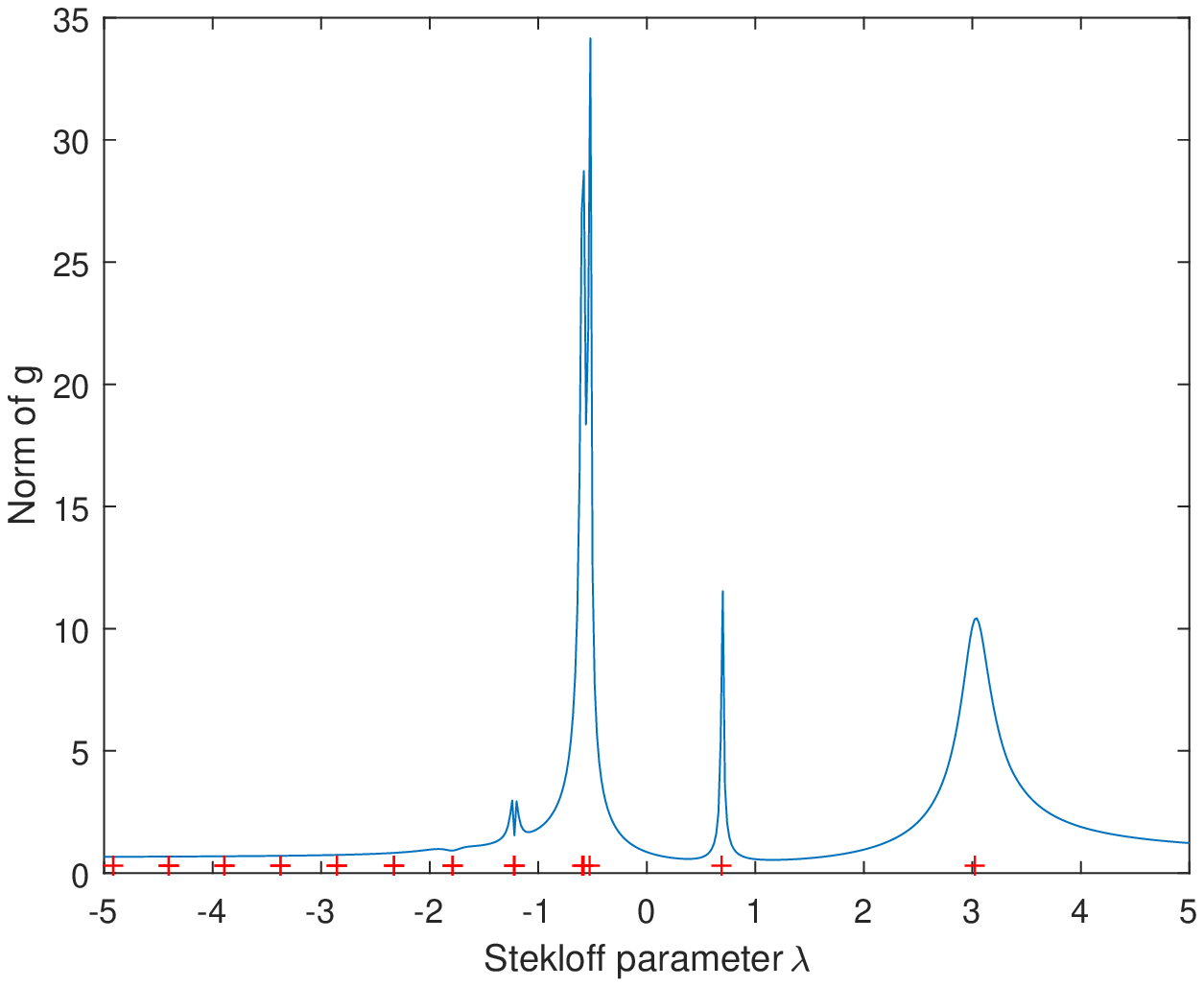}}
\caption{The plots of $|g_\lambda|$ against $\lambda$ for $n(x)=5$.}
\label{fn5}
\end{figure}

{\bf Example 2} Real index of refraction $n(x)=4+2|x|$.
The exact and reconstructed eigenvalues are shown in \ref{n42x}. The plots of $|g_\lambda|$ for three domains are shown in \ref{fn42x}.
\begin{table}[h!]
\centering
\begin{tabular}{l|l|l|l|l}\hline
disc& $2.0856$& $-0.4898$& $-0.5714$& $-1.2285$ \\
 & ($2.10$)& ($-0.48$)& ($-0.54$)& ($-1.20$)\\\hline
square & $0.4137$& $-0.5758$& $-1.2348$&  \\
&($0.42$)& ($-0.58$)& ($-1.22$)& \\\hline
L-shaped & $0.9825$& $-0.4956$& $-0.6018$& $-1.2116$ \\
&($1.00$)& ($-0.50$)& ($-0.60$)& ($-1.20$) \\\hline
\end{tabular}
\caption{\label{n42x} The exact Stekloff eigenvalues and their reconstructions (in the parentheses) for $n(x)=4+2|x|$.}
\end{table}
\begin{figure}[h!]
  \centering
  \subfigure[\textbf{disc}]{
    \includegraphics[width=2.5in]{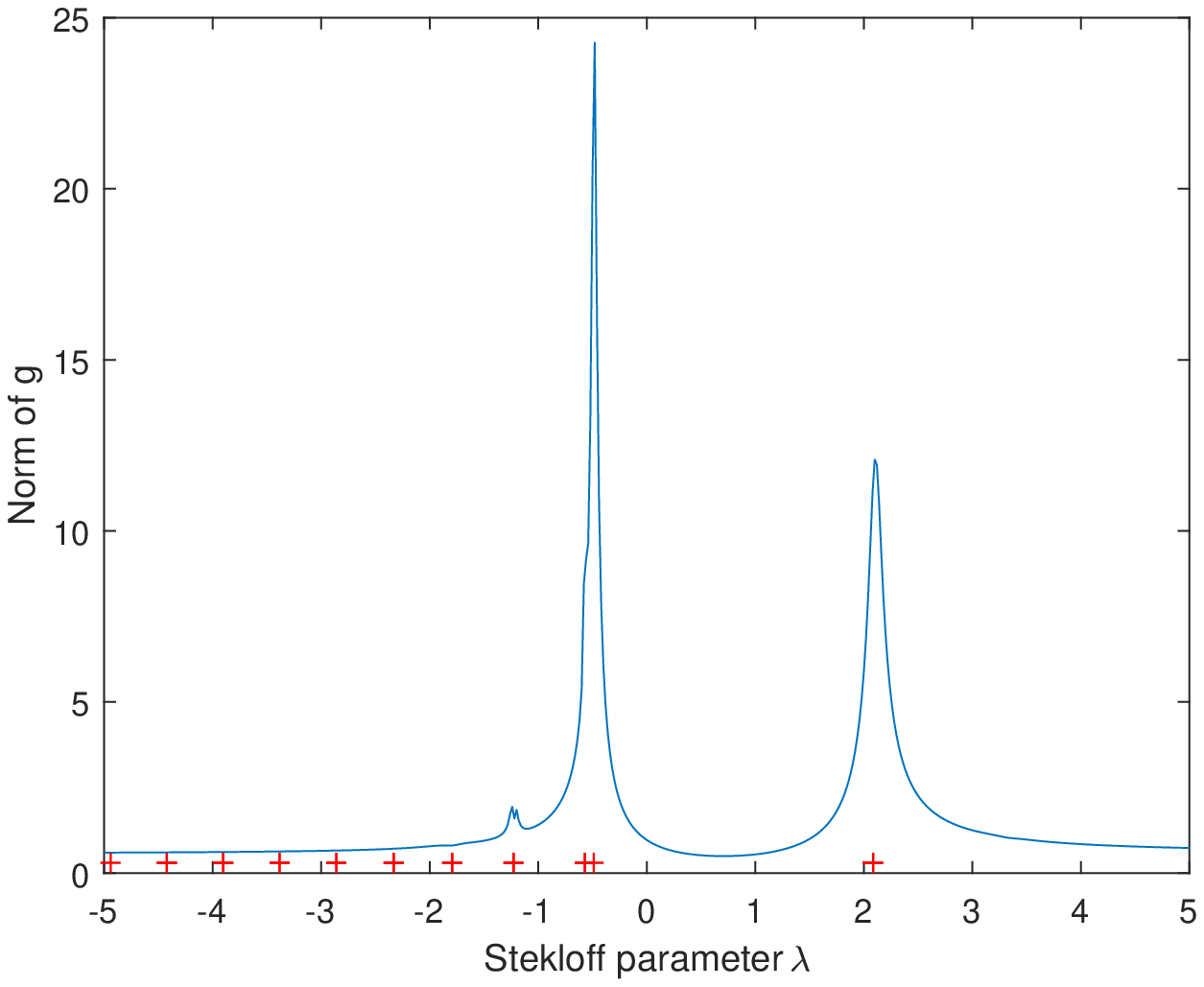}}
  \subfigure[\textbf{square}]{
    \includegraphics[width=2.5in]{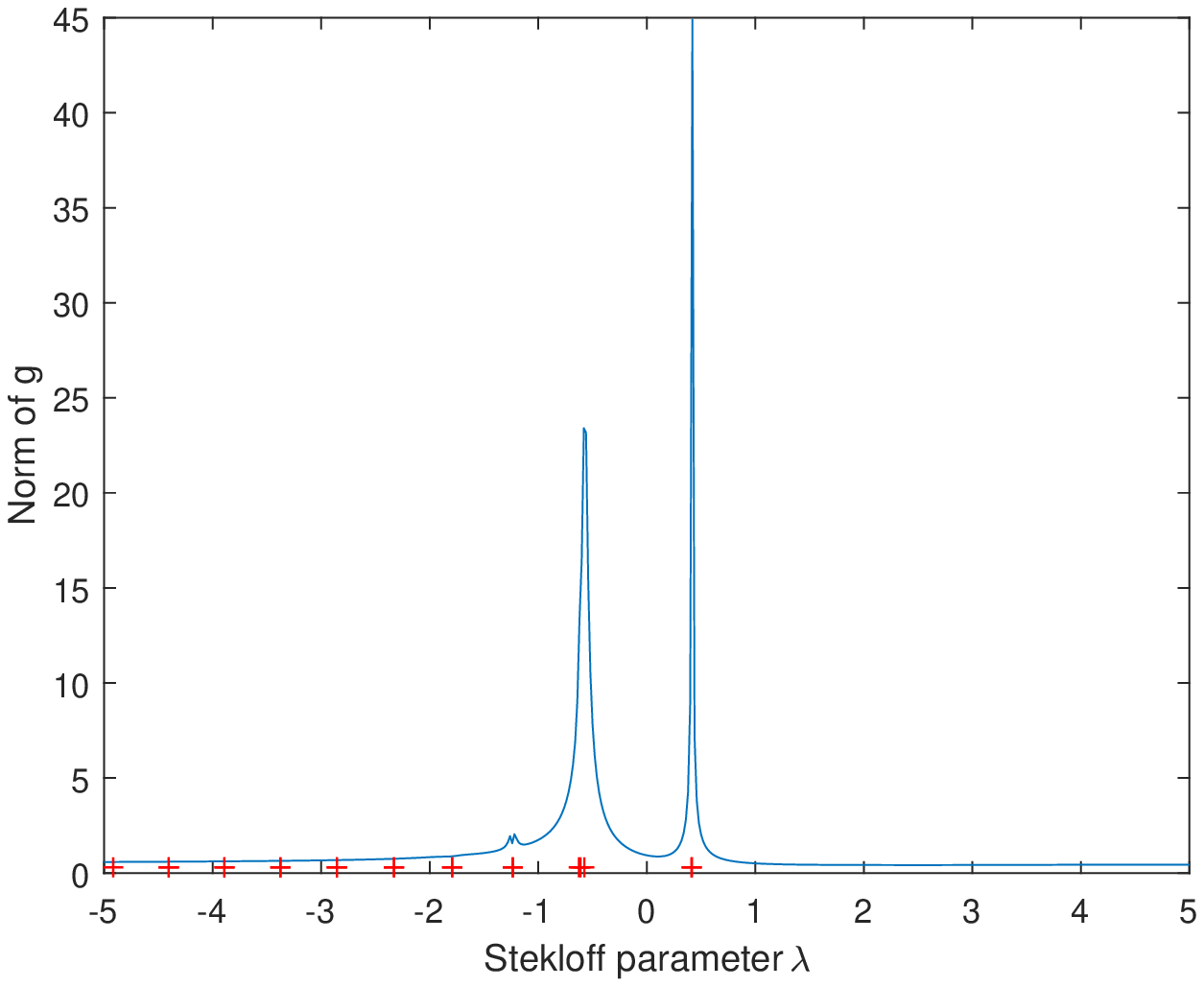}}
  \subfigure[\textbf{L-shaped}]{
    \includegraphics[width=2.5in]{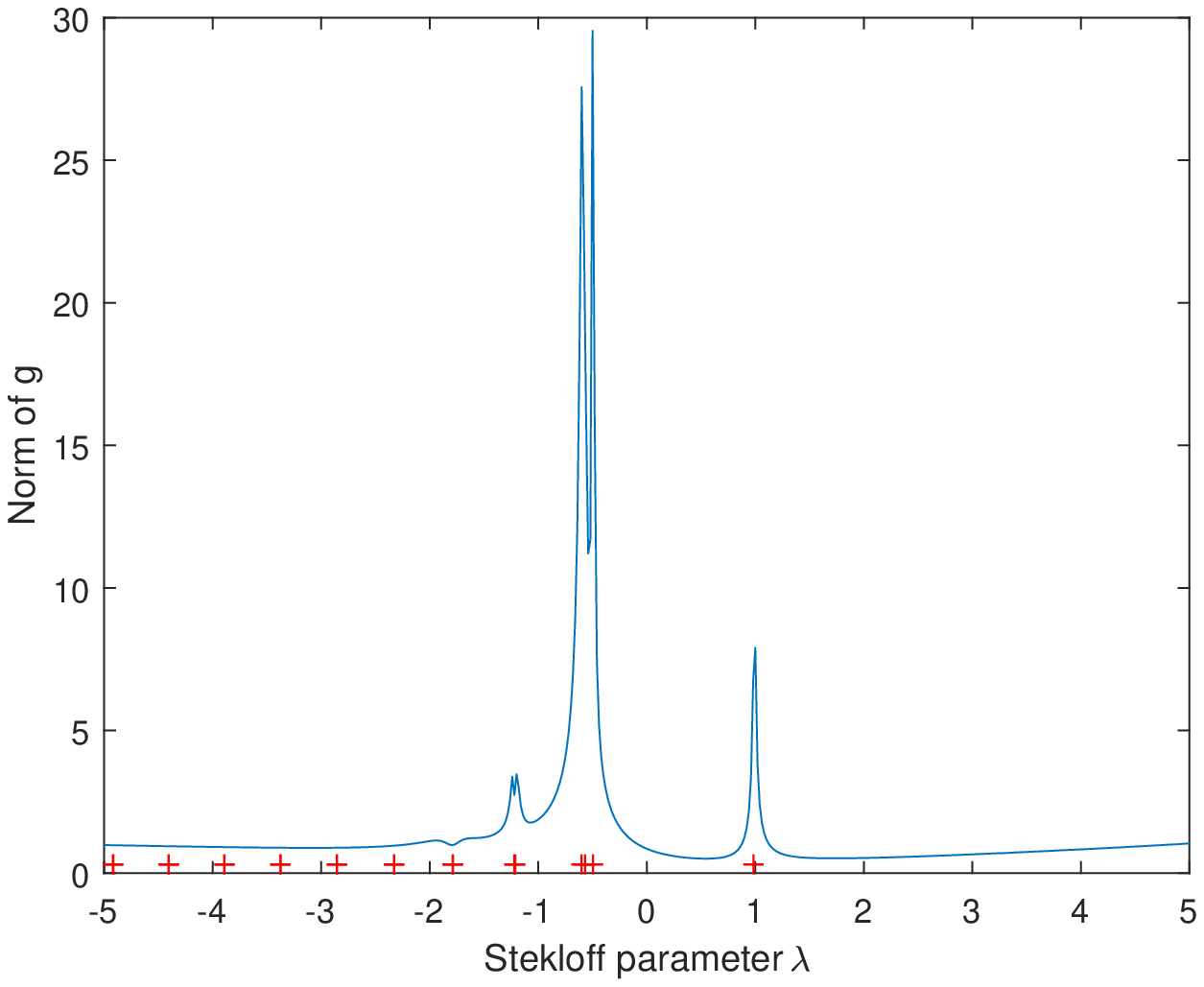}}
\caption{The plots of $|g_\lambda|$ against $\lambda$ for $n(x)=4+2|x|$.}
\label{fn42x}
\end{figure}

\begin{table}[h!]
\centering
\begin{tabular}{l|l|l|l}\hline
disc & $-0.0549 + 0.4854i$& $-0.0211+0.2652i$&$-0.6361 + 0.0390i$ \\
 & ($-0.06+0.46i$) &{-0.02+0.26i}&($-0.64+0.04i$)  \\\hline
square & $0.1303 + 0.5812i$& $0.0779+ 0.1415i$& $-0.6338+0.0170i$   \\
&($0.12+0.56i$) & ($0.08+0.14i$) & ($-0.64+0.02i$) \\\hline
L-shaped &  $-0.0902+ 0.5468i$& $-0.1029+ 0.3600i$& $0.0364 + 0.2228i$\\
 &($-0.10+0.52i$)& ($-0.10+0.36i$)& ($0.04+0.22i$)\\
& $-0.6385 + 0.0764i$& $-0.6346 + 0.0412i$ & \\
& ($-0.64+0.08i$)&($-0.64+0.04i$) & \\\hline
\end{tabular}
\caption{\label{n24i} The exact Stekloff eigenvalues and their reconstructions (in the parentheses) for $n(x)=2+4i$.}
\end{table}

{\bf Example 3} Complex index of refraction $n(x)=2+4i$.
The exact and reconstructed eigenvalues are shown in \ref{n24i}. The plots of $|g_\lambda|$ for three domains are shown in \ref{fn24iS}.
\begin{figure}[h!]
  \centering
  \subfigure[\textbf{disc}]{
    \includegraphics[width=2.5in]{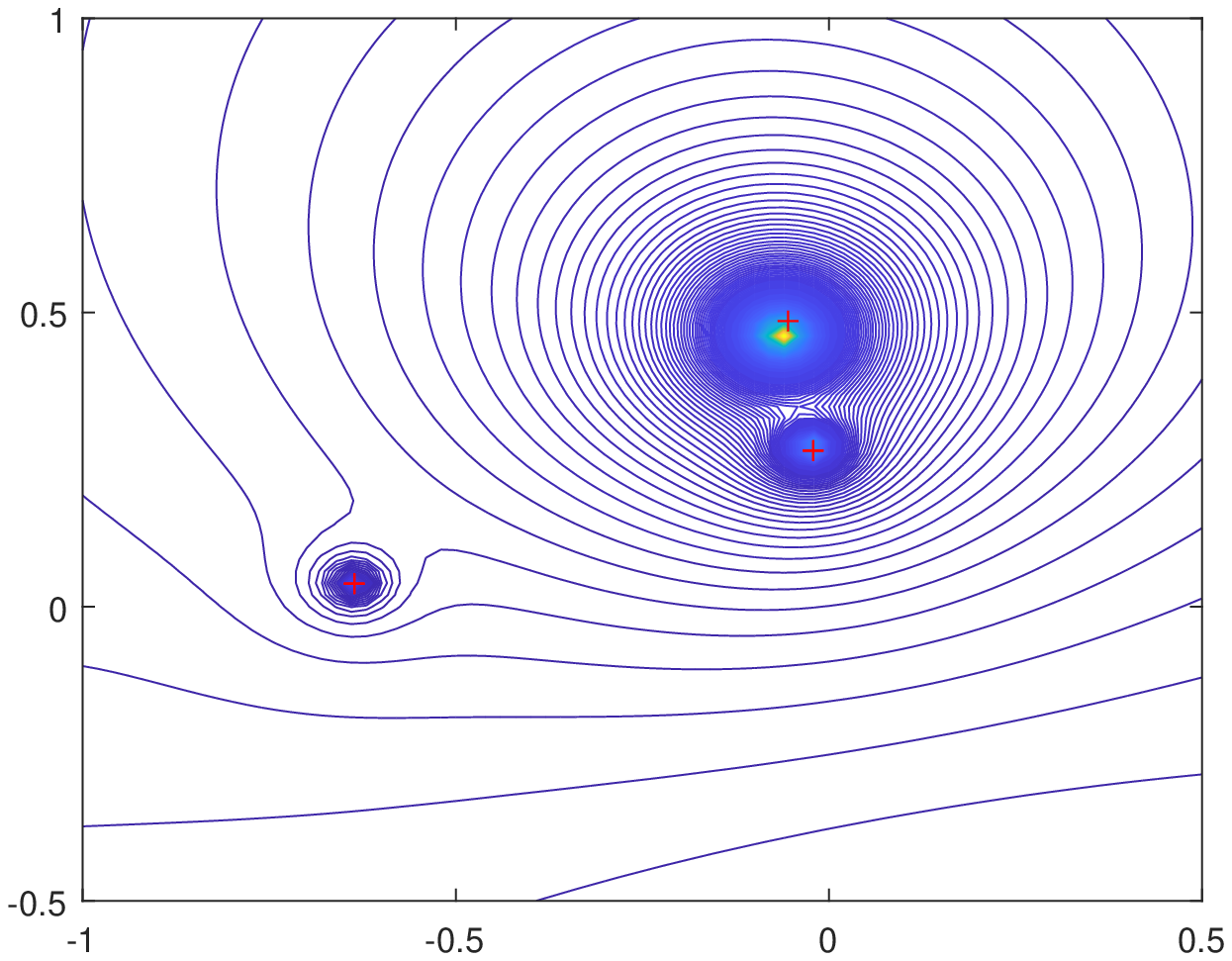}}
  \subfigure[\textbf{square}]{
    \includegraphics[width=2.5in]{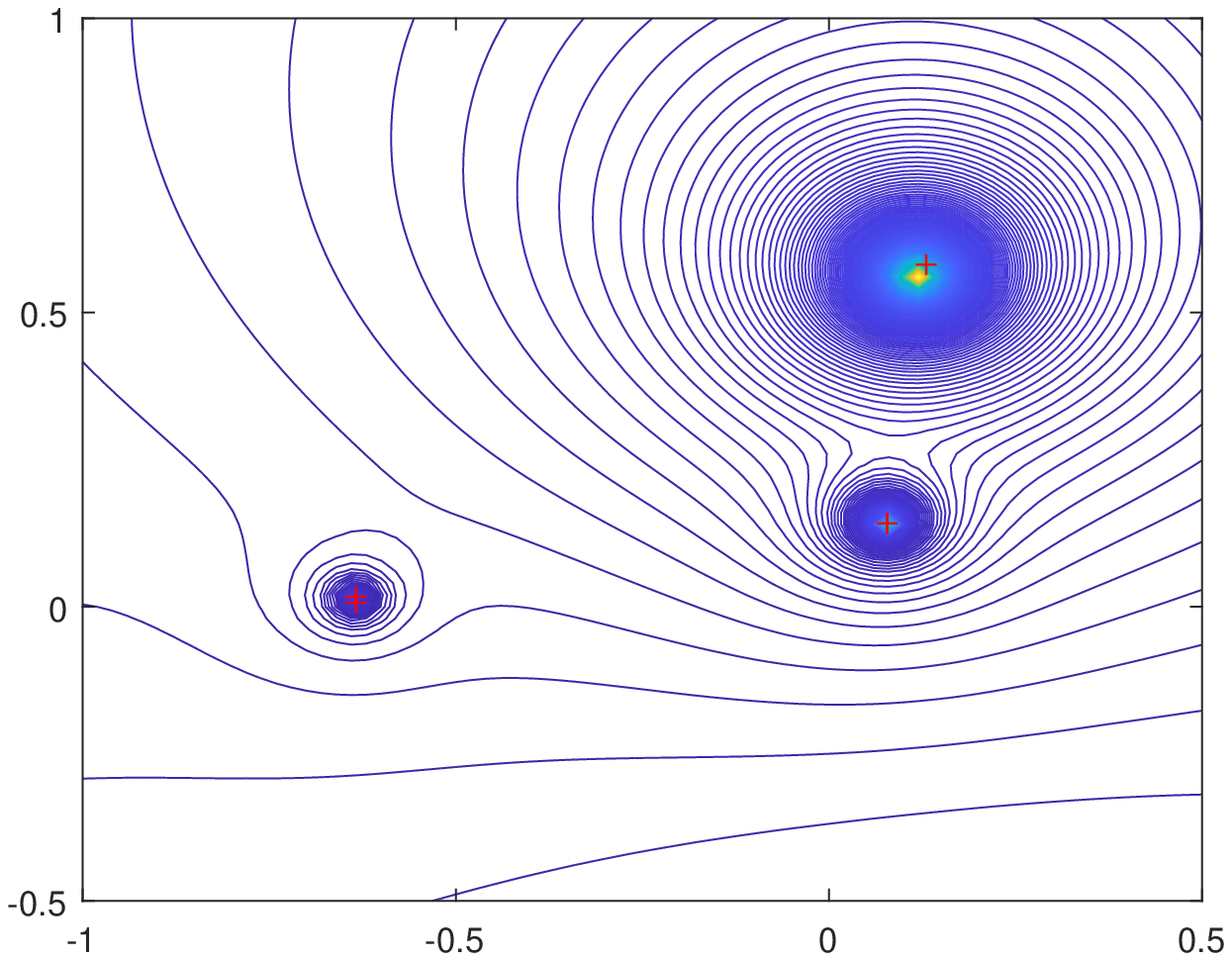}}
  \subfigure[\textbf{L-shaped}]{
    \includegraphics[width=2.5in]{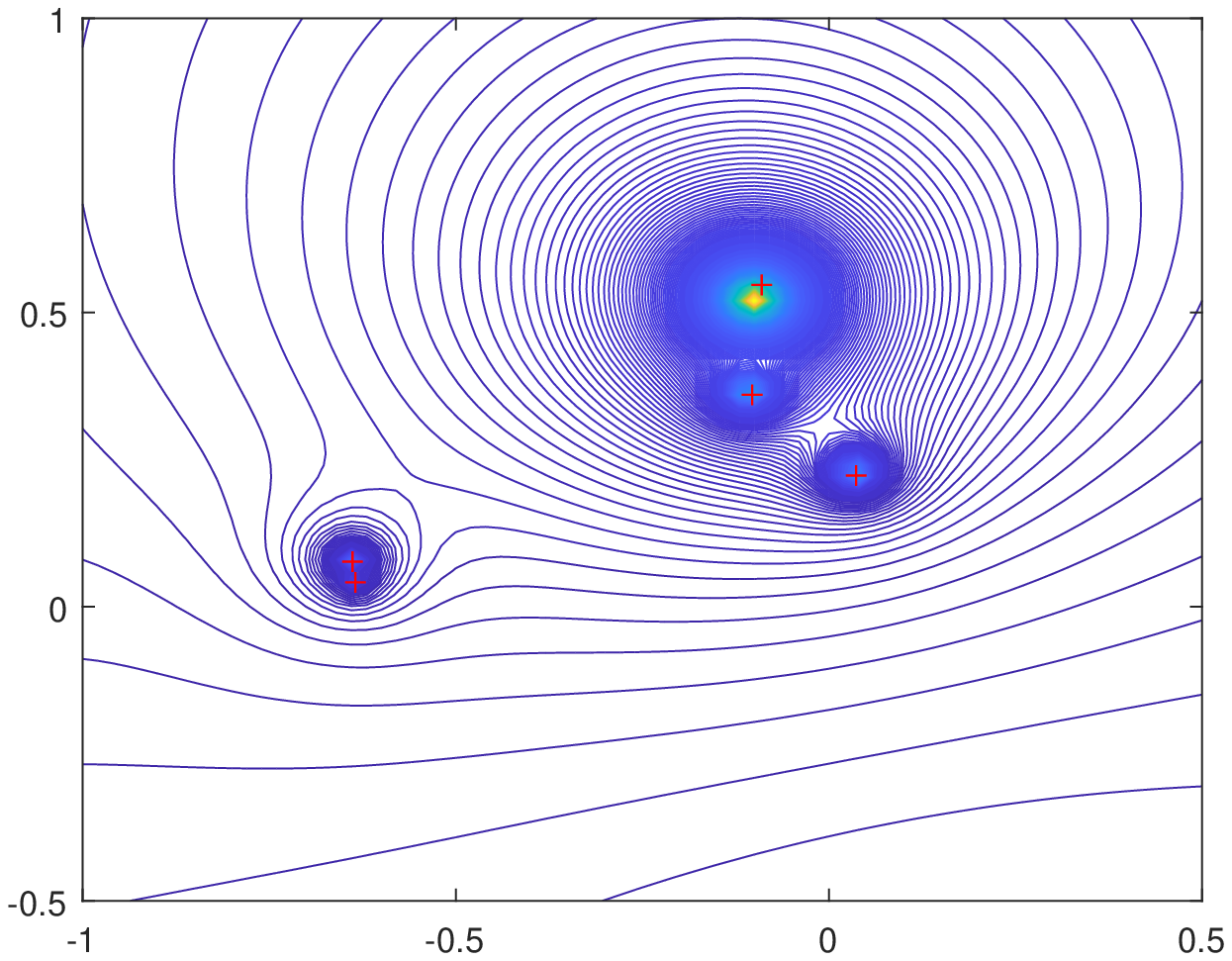}}
\caption{The plots of $|g_\lambda|$ against $\lambda$ for $n(x)=2+4i$.}
\label{fn24iS}
\end{figure}

\subsection{Estimation of the index of refraction}
Given the reconstructed Stekloff eigenvalues,
we present some numerical examples for the estimation of the index of refraction by the Bayesian approach.
The examples are rather naive. Nonetheless, the results show the potential of statistical approaches for inverse scattering problems.
Since the main goal is to show the effectiveness of the Bayesian approach, we assume that the shape of the scatterer is known in the following examples.

{\bf Example 4} Real constant index of refraction $n(x)=5$.
Assume one Stekloff eigenvalue is reconstructed from Cauchy data:  $-0.48$ for the disk, $-0.54$ for the square, and $-0.52$ for the L-shaped domain.
Since $n$ is a real constant, we take a uniform prior $\mathcal{U}(0,8)$. The posterior density is given by
\begin{equation}
\pi_{post}(n|\lambda)\propto \exp \Big(-\frac{1}{2 \sigma^{2} }\|\lambda-\mathcal{G}(n)\|^2\Big)\times I\{0<n<8\}.
\end{equation}
We generate 3000 samples for each domain. The initial sample is chosen to be $n_{1}=2$.
The rest samples $\{n_{i}\}_{i=2}^{3000}$ are drawn from the symmetric proposal distribution
\[
q(n_{j},n_{j-1})\propto \exp \Big( -\frac{1}{2 \gamma ^2}\|n_{j}-n_{j-1}\|^2 \Big),
\]
where $\gamma^{2}=2.4^2/2$.

\ref{table1} shows $n_{CM}$ for three domains. The Markov chains are shown in \ref{SE1n5}.
The samples concentrate around $n=5$ for the unit circle domain and square domain.
However, for the L-shaped domain, the samples are accumulated around two values, $5$ and $7$. In fact, this implies that one Stekloff eigenvalue
cannot uniquely determine the constant index of refraction.
If two Stekloff eigenvalues, $0.70$ and $-0.52$, are used (see \ref{SE2n5L}), we obtain $n_{CM}=5.0074$, which is a good approximation of $5$.
\begin{table}[h]
\centering
\begin{tabular}{c|c|c }
\hline
  $\sigma^2$ \qquad & domain\qquad & $n_{CM}$\\\hline
\multicolumn{1}{c|}{\multirow{3}{*}{$0.05$}} &circle &  4.9953
\\
    \multicolumn{1}{c|}{} & square &     5.0205  \\
  \multicolumn{1}{c|}{} & L-shaped domain &  6.2135 (5.0074) \\
  \hline
\end{tabular}
\caption{\label{table1} The posterior means for three domains ($n(x)=5$) using one Stekloff eigenvalue.
The value in the parentheses for the L-shaped domain is obtained using two Stekloff eigenvalues.}
\end{table}

\begin{figure}[h!]
  \centering
  \subfigure[\textbf{disc}]{
    \includegraphics[width=2.6in]{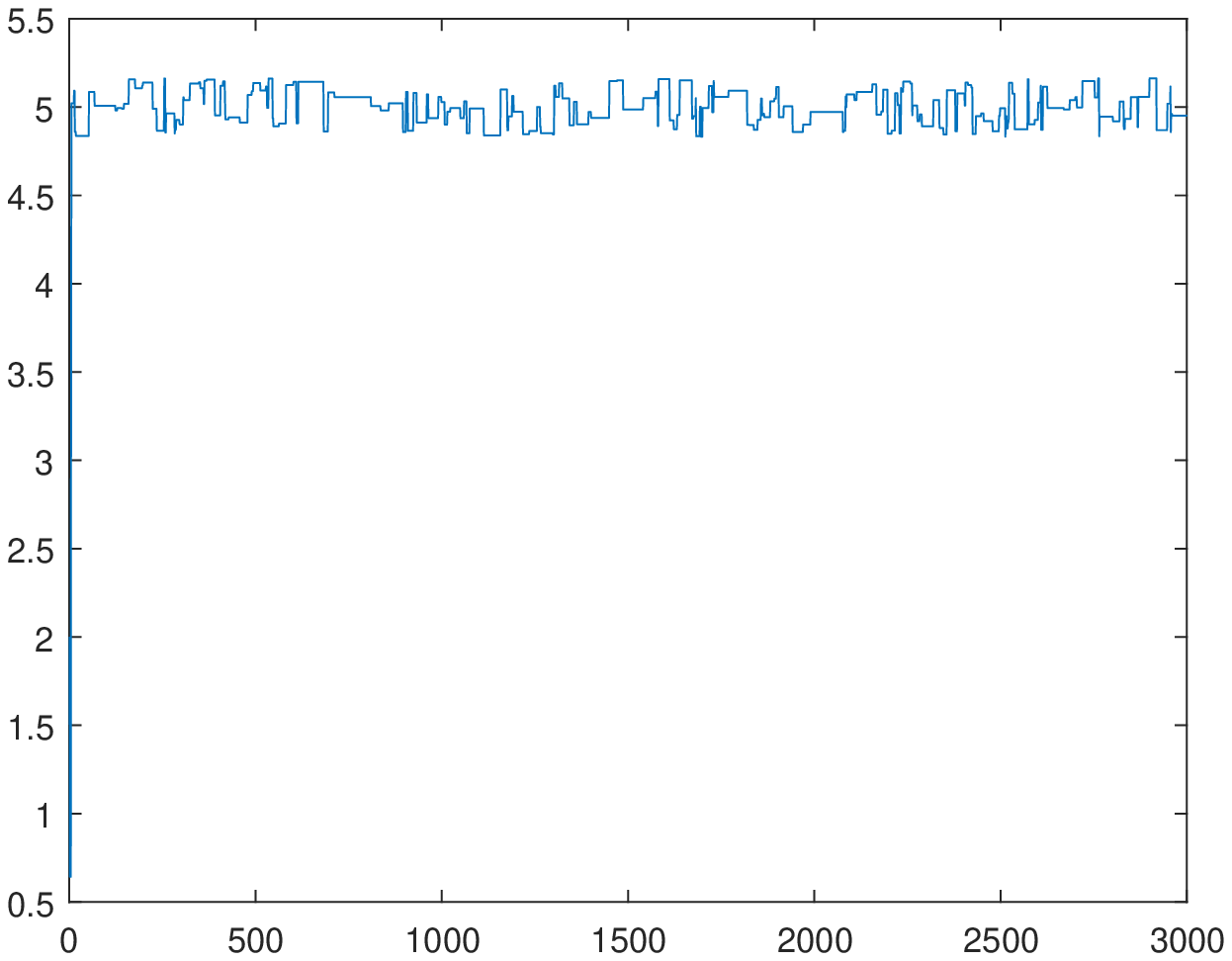}}
  \subfigure[\textbf{disc}]{
    \includegraphics[width=2.8in]{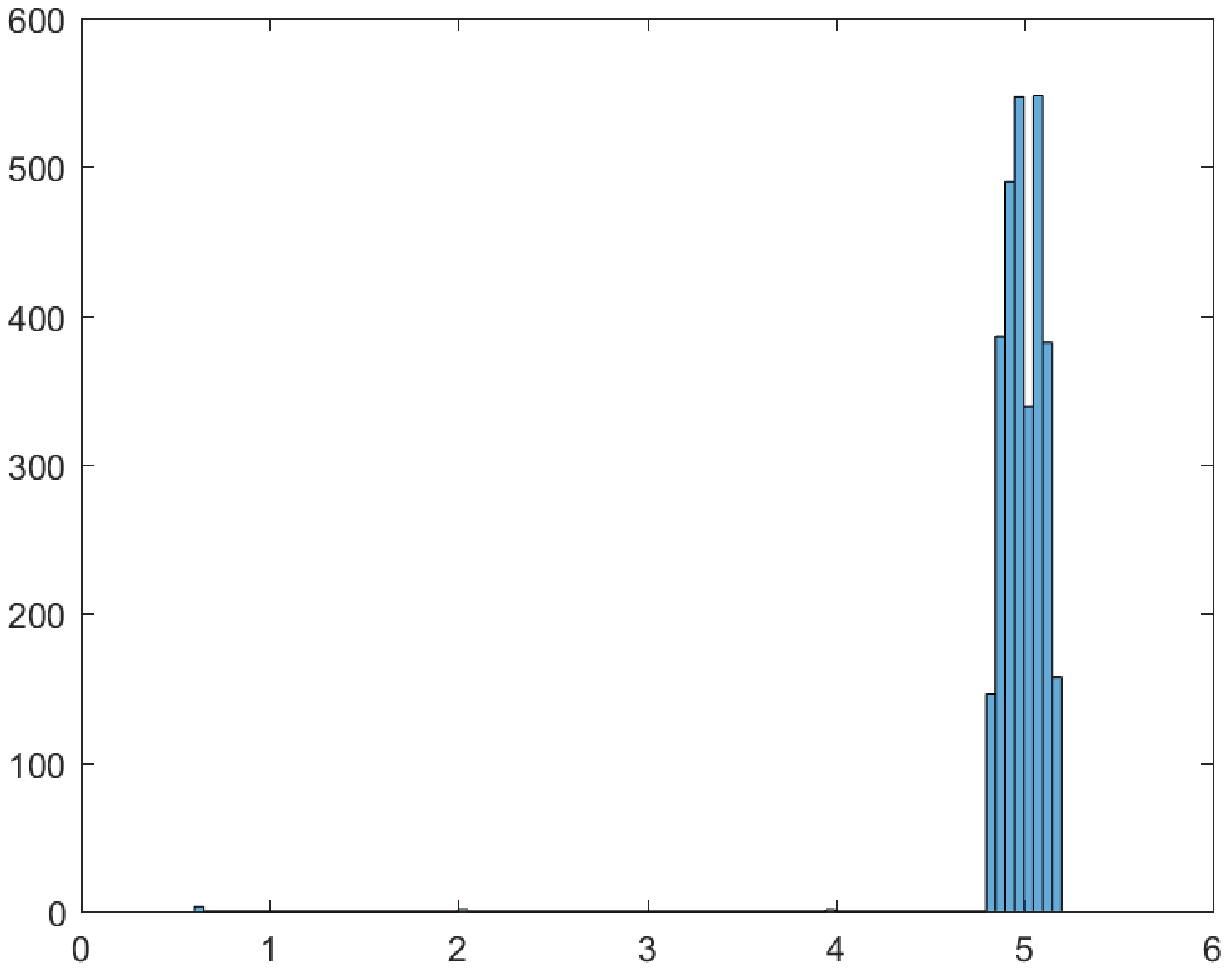}}
   \subfigure[\textbf{square}]{
    \includegraphics[width=2.6in]{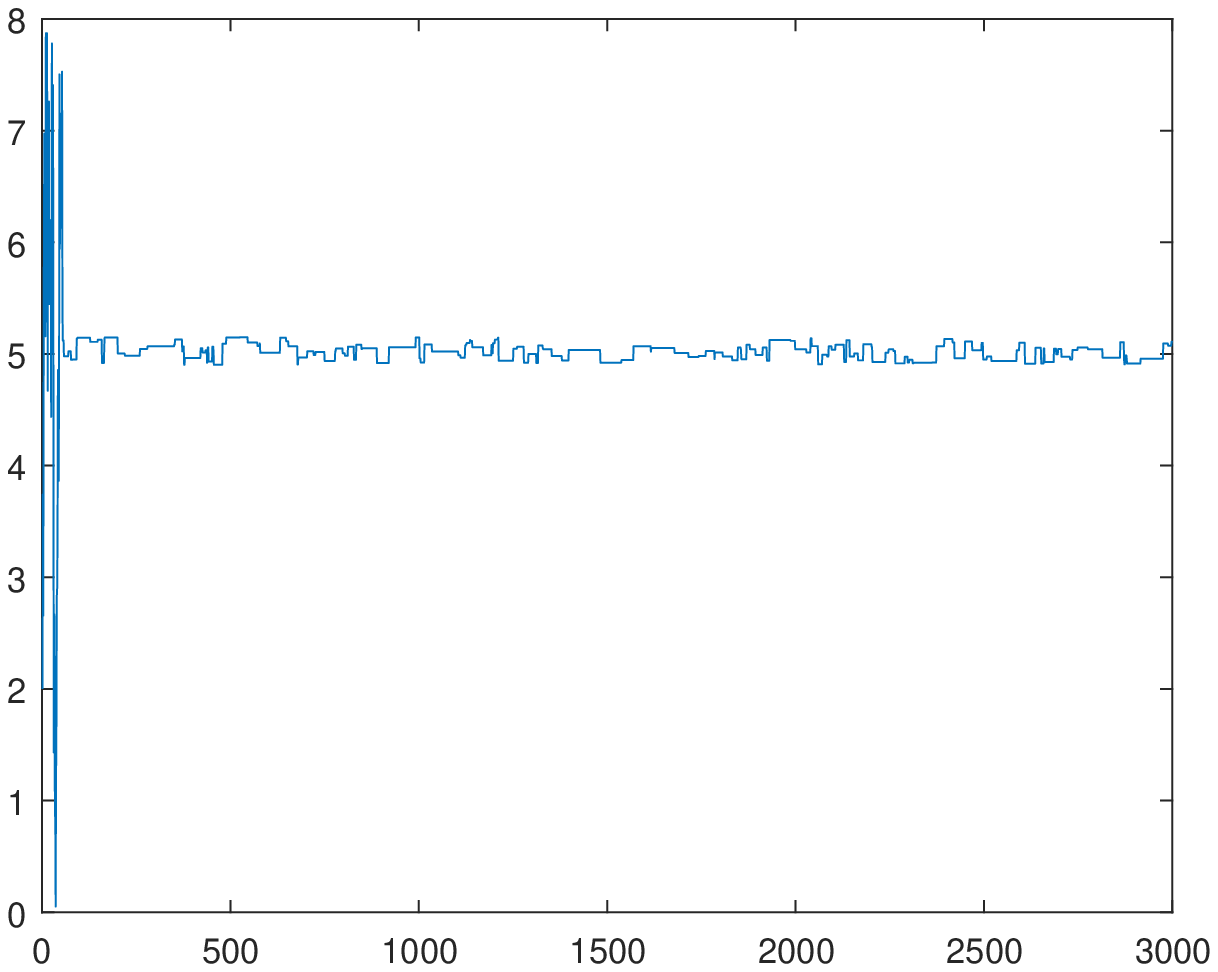}}
  \subfigure[\textbf{square}]{
    \includegraphics[width=2.8in]{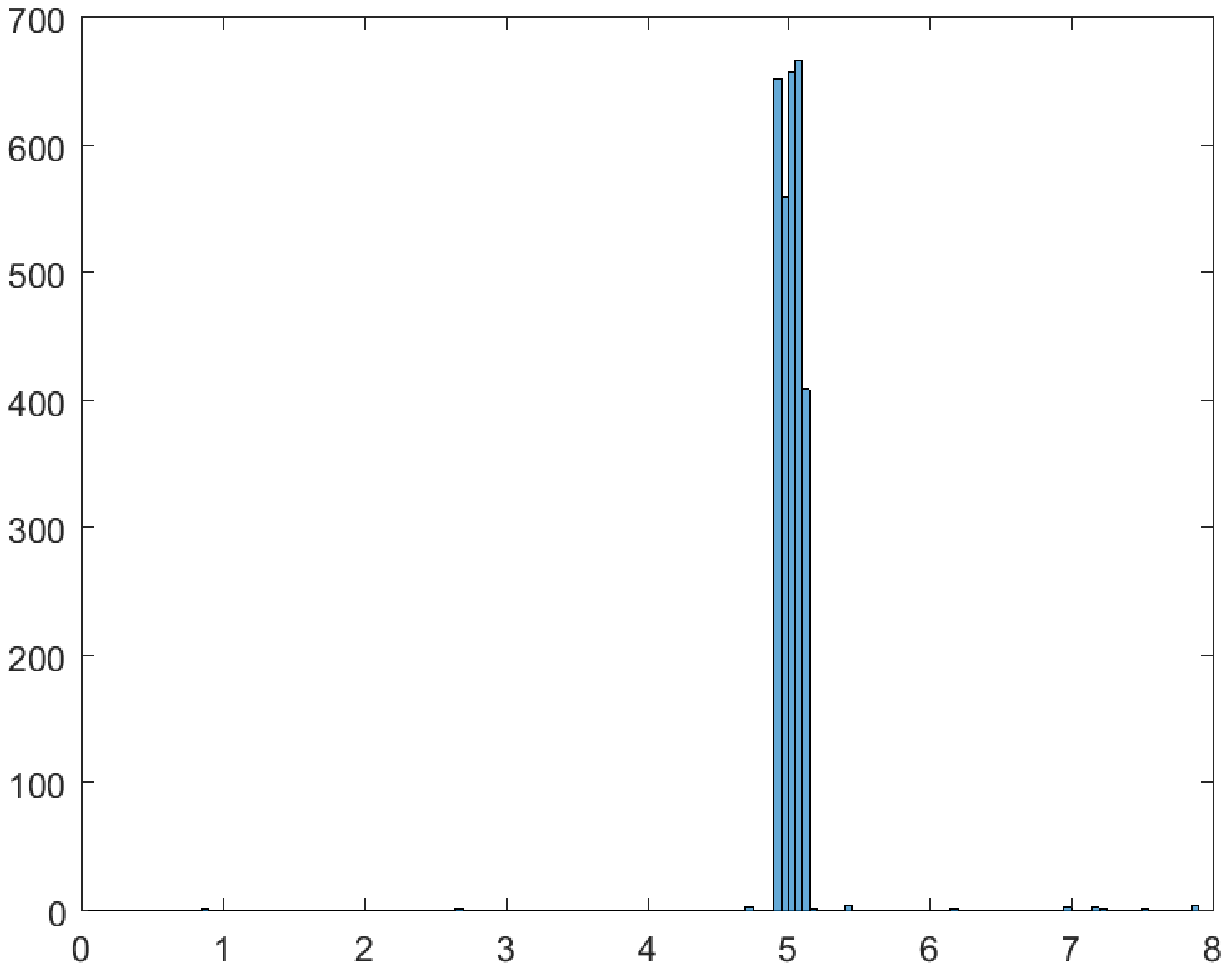}}
  \subfigure[\textbf{L-shaped}]{
    \includegraphics[width=2.6in]{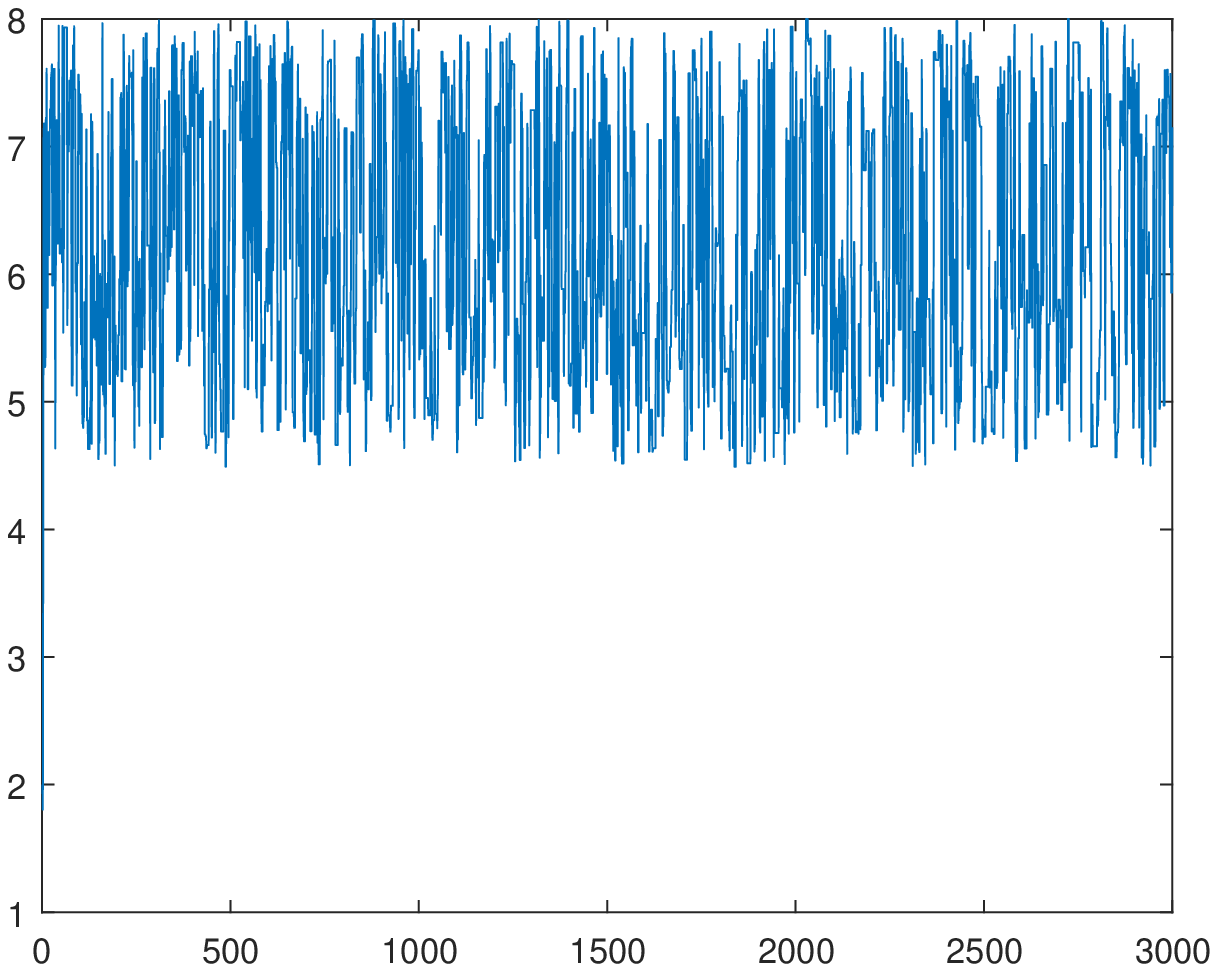}}
  \subfigure[\textbf{L-shaped}]{
    \includegraphics[width=2.8in]{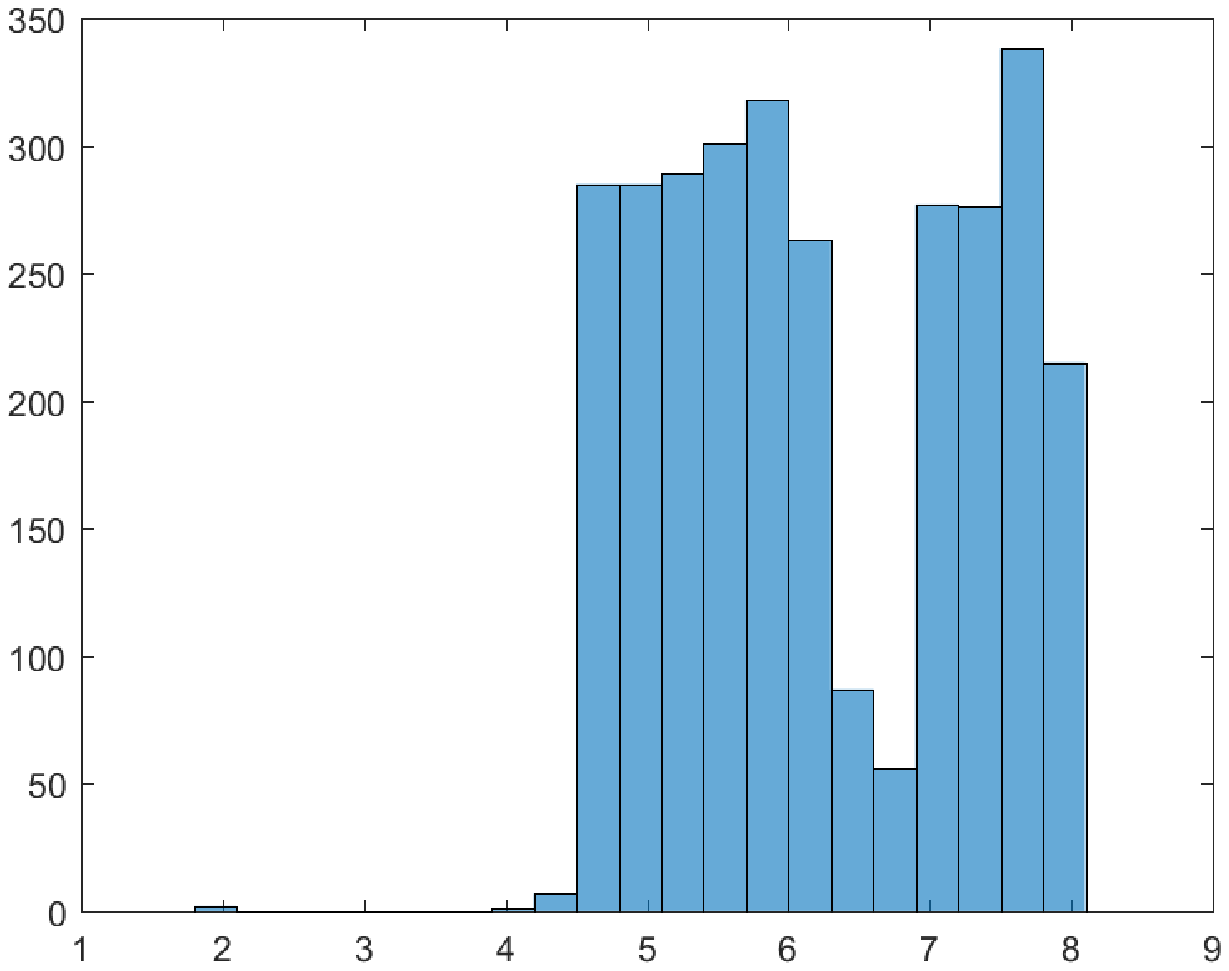}}
\caption{Left: Trace plots and histograms of Markov chains for three domains when $n\equiv 5$. Right: Probability histograms.}
\label{SE1n5}
\end{figure}

\begin{figure}[h!]
  \centering
  \subfigure[\textbf{L-shaped}]{
    \includegraphics[width=2.6in]{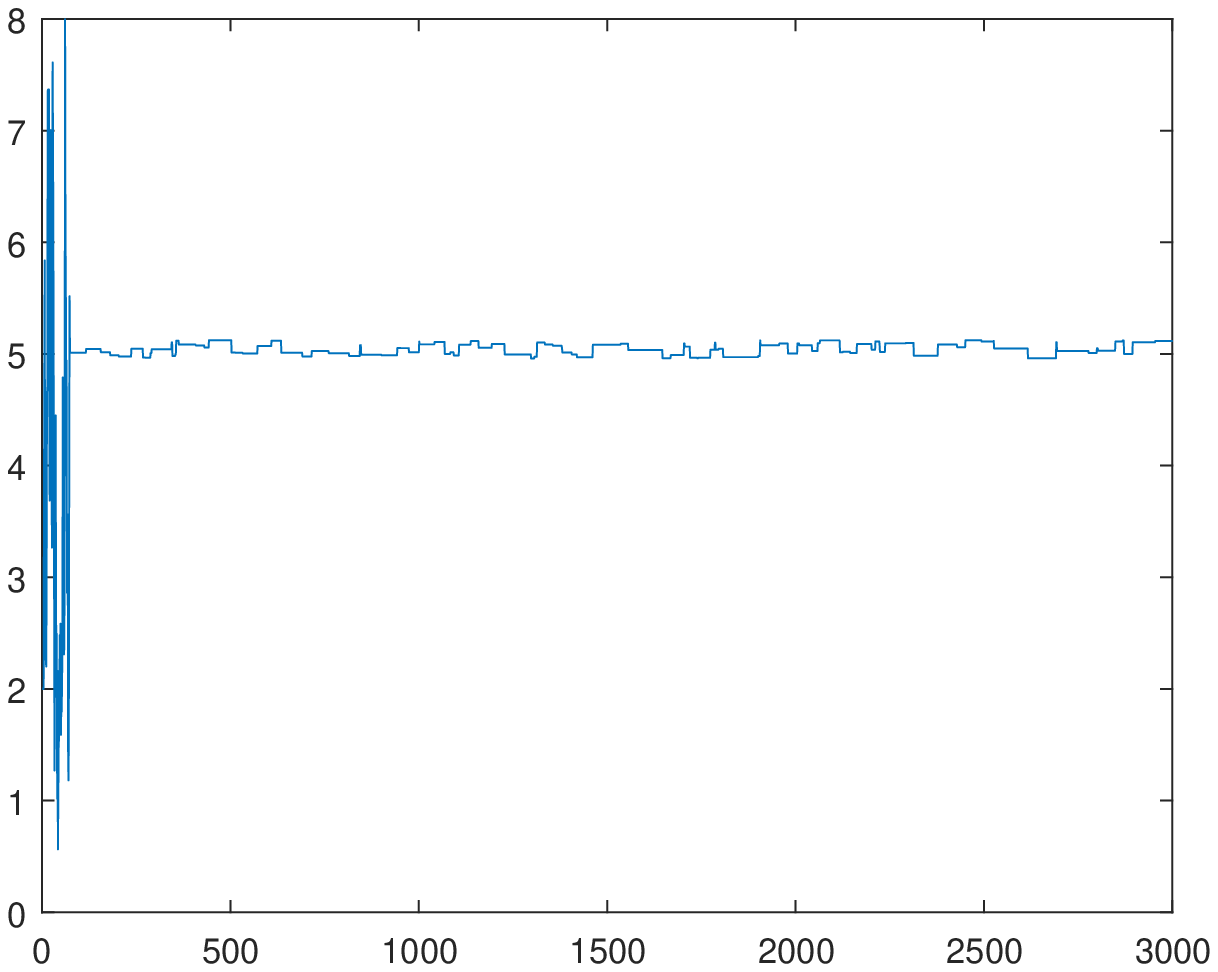}}
  \subfigure[\textbf{L-shaped}]{
    \includegraphics[width=2.8in]{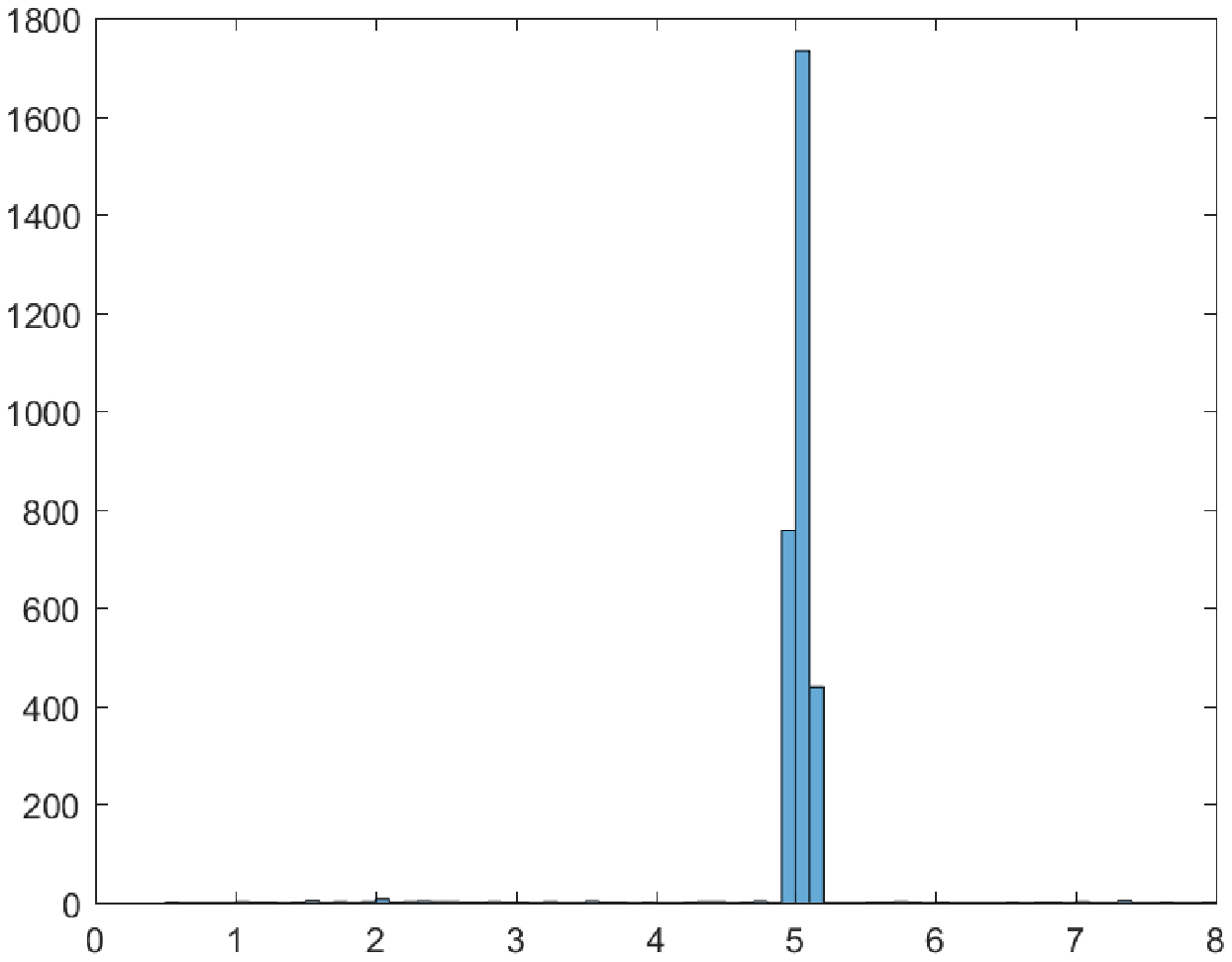}}
\caption{Left: Trace plot of the Markov Chain for the L-shaped using two eigenvalues ($n=5$). Right: Probability histogram.}
\label{SE2n5L}
\end{figure}

{\bf Example 5}  Real function index of refraction $n(x)=4+2|x|$.
Assume that two Stekloff eigenvalues are given and $n(x)$ is of the form $\beta^1+\beta^2|x|$.
We first obtain a constant approximation $n_{0}$ for $n(x)$ as the above example.
This provides some ideas of how to choose priors for $\beta^1$ and $\beta^2$.
For the second step, two Stekloff eigenvalues are used.
The posterior distribution is given by
\begin{equation} 
\pi(n|{\boldsymbol \lambda})\propto \exp \Big(-\frac{1}{2 \sigma^{2} }\|{\boldsymbol \lambda}-\mathcal{G}(n)\|^2\Big)\times I\{3<\beta^{1}<7\} \times I\{0<\beta^{2}<6\}.
\end{equation}

Two reconstructed Stekloff eigenvalues from \ref{n42x} are used for each domain: $2.10, -0.48$ for disc,  $0.42, -0.58$ for the square, and $1.00, -0.50$ for the
L-shaped domain.
\ref{table2} shows the reconstruction results.

\begin{table}[h!]
\centering
\begin{tabular}{c|c|c|c}
\hline
  $\sigma^2$ \qquad & domain\qquad &  \text{mean of} $n_{0}$ &  $\beta_{CM}^{1}+\beta_{CM}^{2}|x|$  \\\hline
\multicolumn{1}{c|}{\multirow{3}{*}{$0.05$}} &circle &  $4.9804$& $4.3916+1.9333|x|$\\
    \multicolumn{1}{c|}{} & square &   $6.7791$&   $3.8873+ 2.3409|x|$\\
  \multicolumn{1}{c|}{} & L-shaped &  $6.2197$ & $4.2953+1.7263|x|$ \\
  \hline
\end{tabular}
\caption{\label{table2} The posterior mean of $n(x)$ for three domains ($n(x)=4+2|x|$).}
\end{table}

{\bf Example 6} Complex index of refraction $n(x)=2+4i$. Assume that $\Re( n) \sim \mathcal{U}(0,8)$ and $\Im( n)\sim \mathcal{U}(0,8)$.
The same proposal distribution $q(n_{j},n_{j-1})$ are used to sample both $\Re( n)$ and $\Im( n)$.
We use Stekloff eigenvalues, $-0.02+0.26i$ for the circle, $-0.64+0.02i,0.12+0.56i$ for the square, and $-0.1+0.52i,0.04+0.22i$ for the L-shaped domain. 
\ref{table3} shows the reconstruction results and \ref{fn24i} shows the Markov chains.

\begin{table}[h!]
\centering
\begin{tabular}{c|c|c }
\hline
  $\sigma^2$ \qquad & domain\qquad & $n_{CM}$   \\\hline
\multicolumn{1}{c|}{\multirow{3}{*}{$0.05$}} &circle & $1.8511+3.9849i$
\\
    \multicolumn{1}{c|}{} & square & $2.2515+4.1935i$ \\
  \multicolumn{1}{c|}{} & L-shaped & $2.1204+4.1978i$\\
  \hline
\end{tabular}
\caption{\label{table3} The posterior means of $n(x)$ for three domains ($n(x)=2+4i$).}
\end{table}

\begin{figure}[h!]
  \centering
  \subfigure[\textbf{disc (real)}]{
    \includegraphics[width=2.5in]{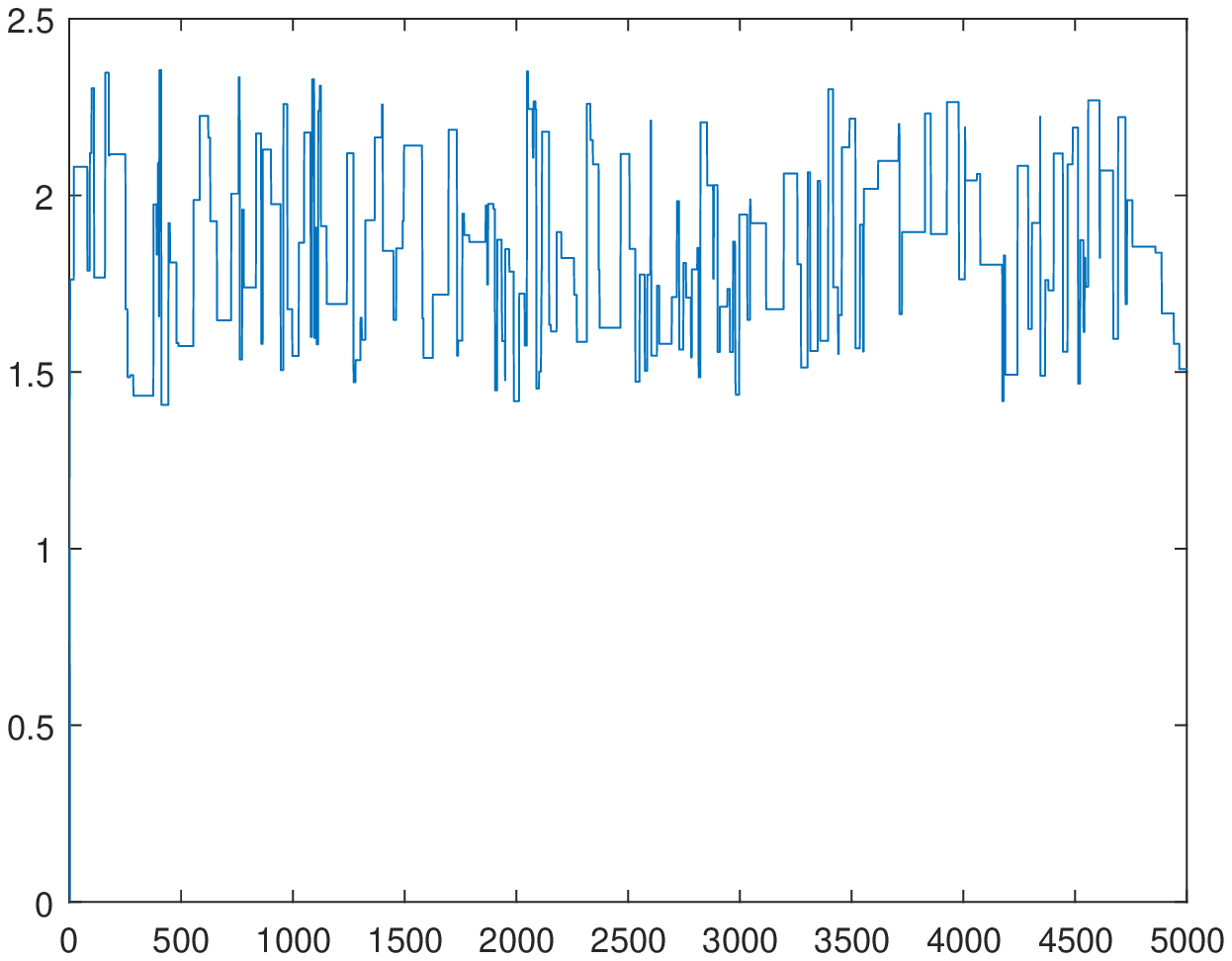}}
  \subfigure[\textbf{disc (imag)}]{
    \includegraphics[width=2.5in]{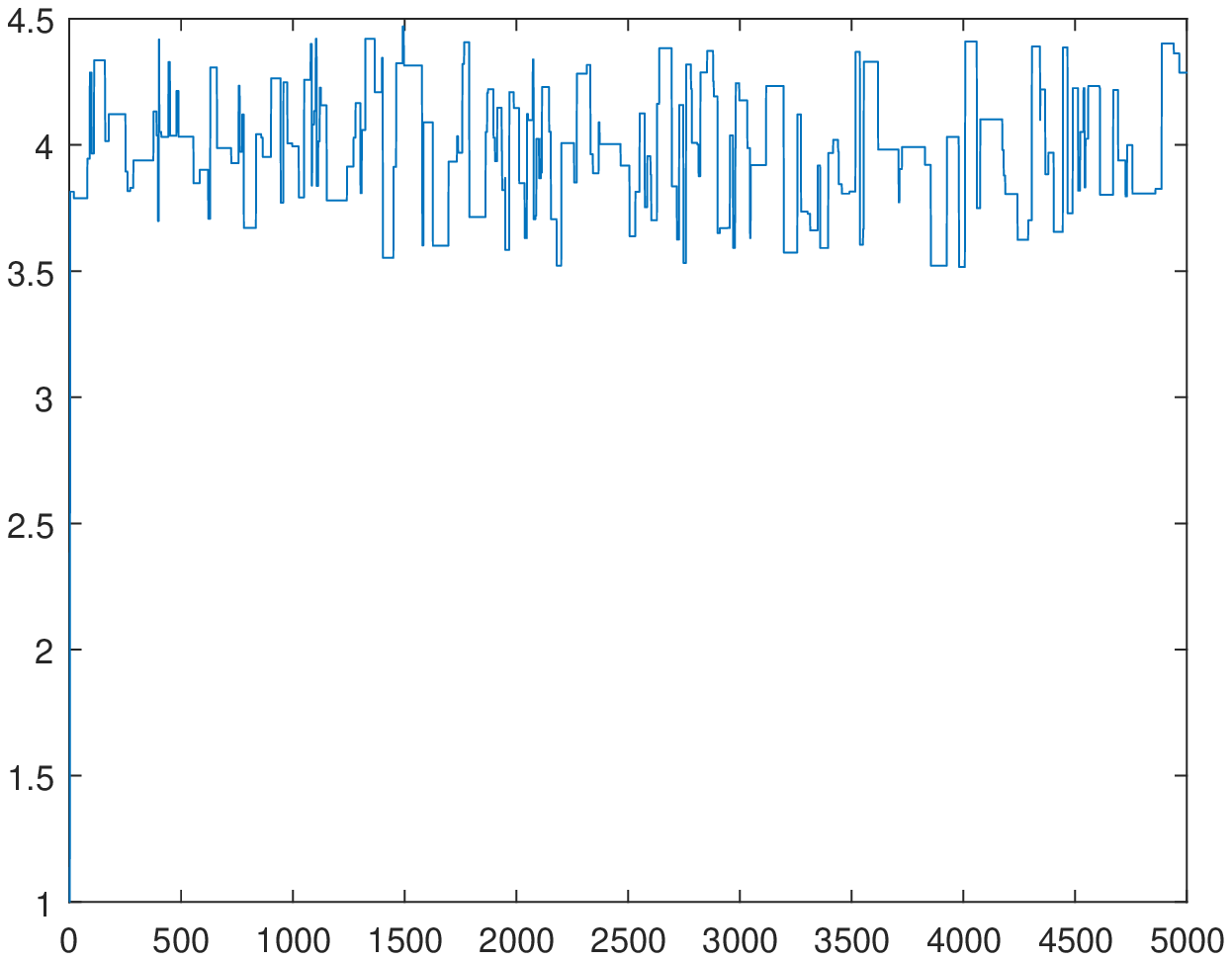}}
   \subfigure[\textbf{square (real)}]{
    \includegraphics[width=2.5in]{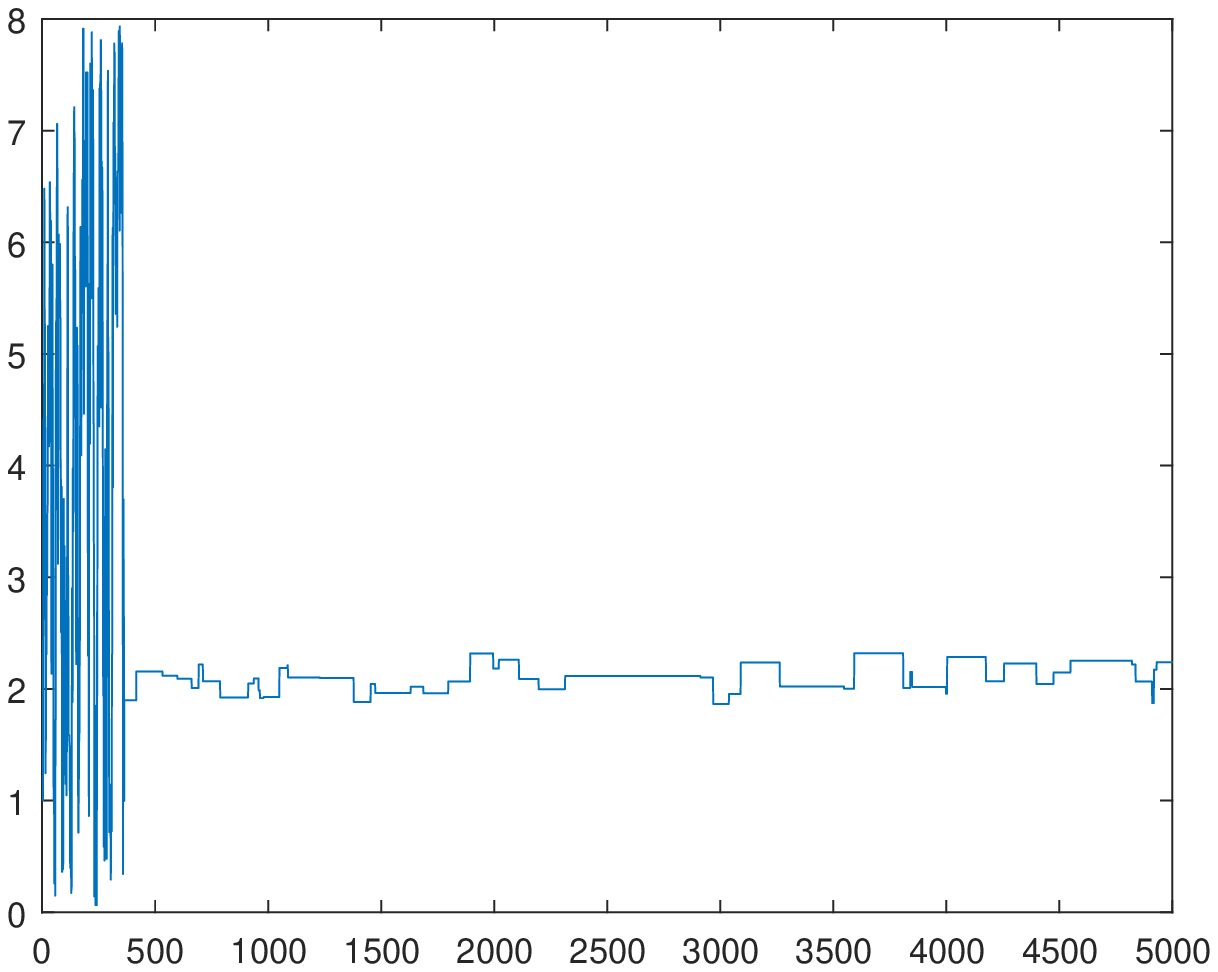}}
  \subfigure[\textbf{square (imag)}]{
    \includegraphics[width=2.5in]{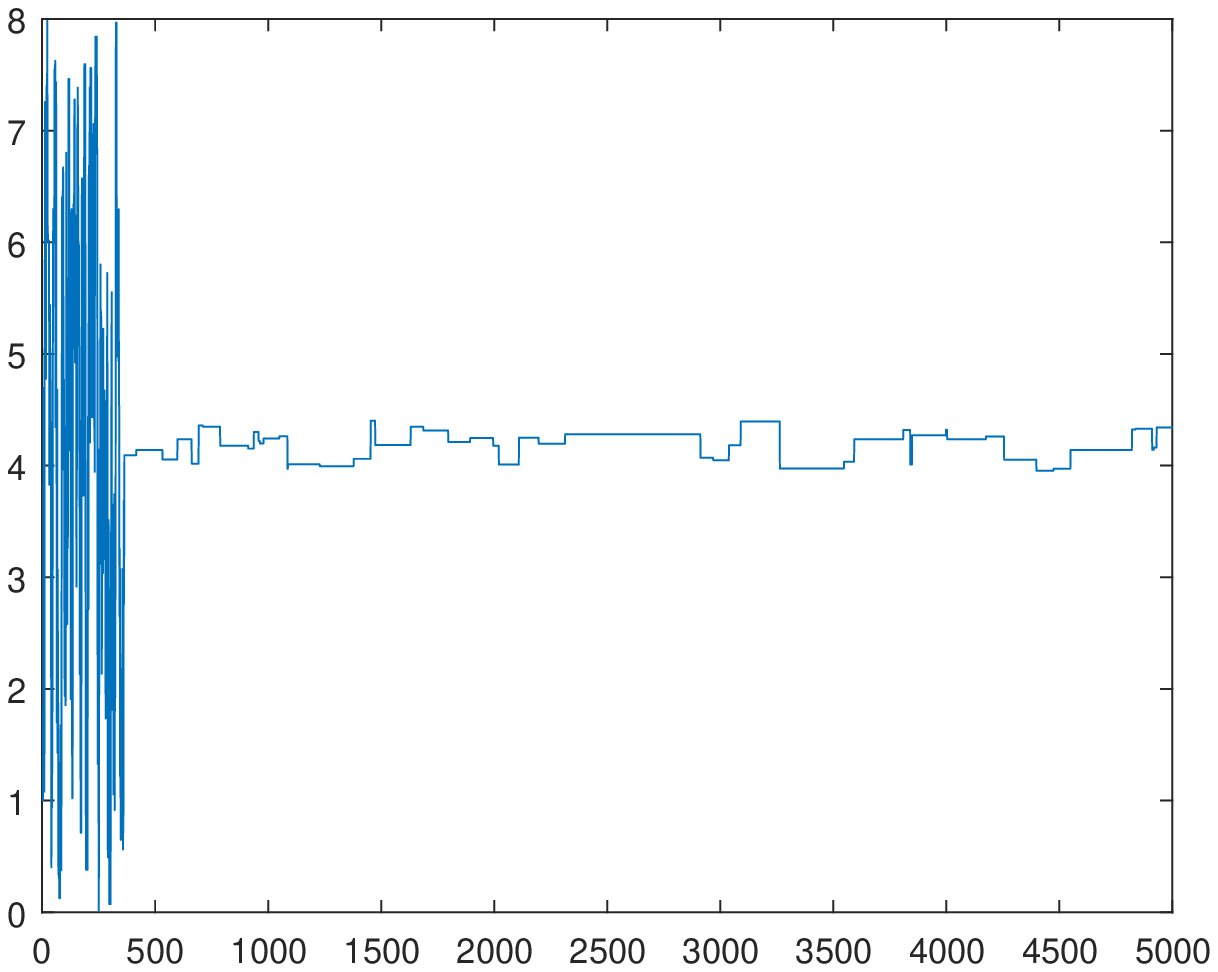}}
  \subfigure[\textbf{L-shaped (real)}]{
    \includegraphics[width=2.5in]{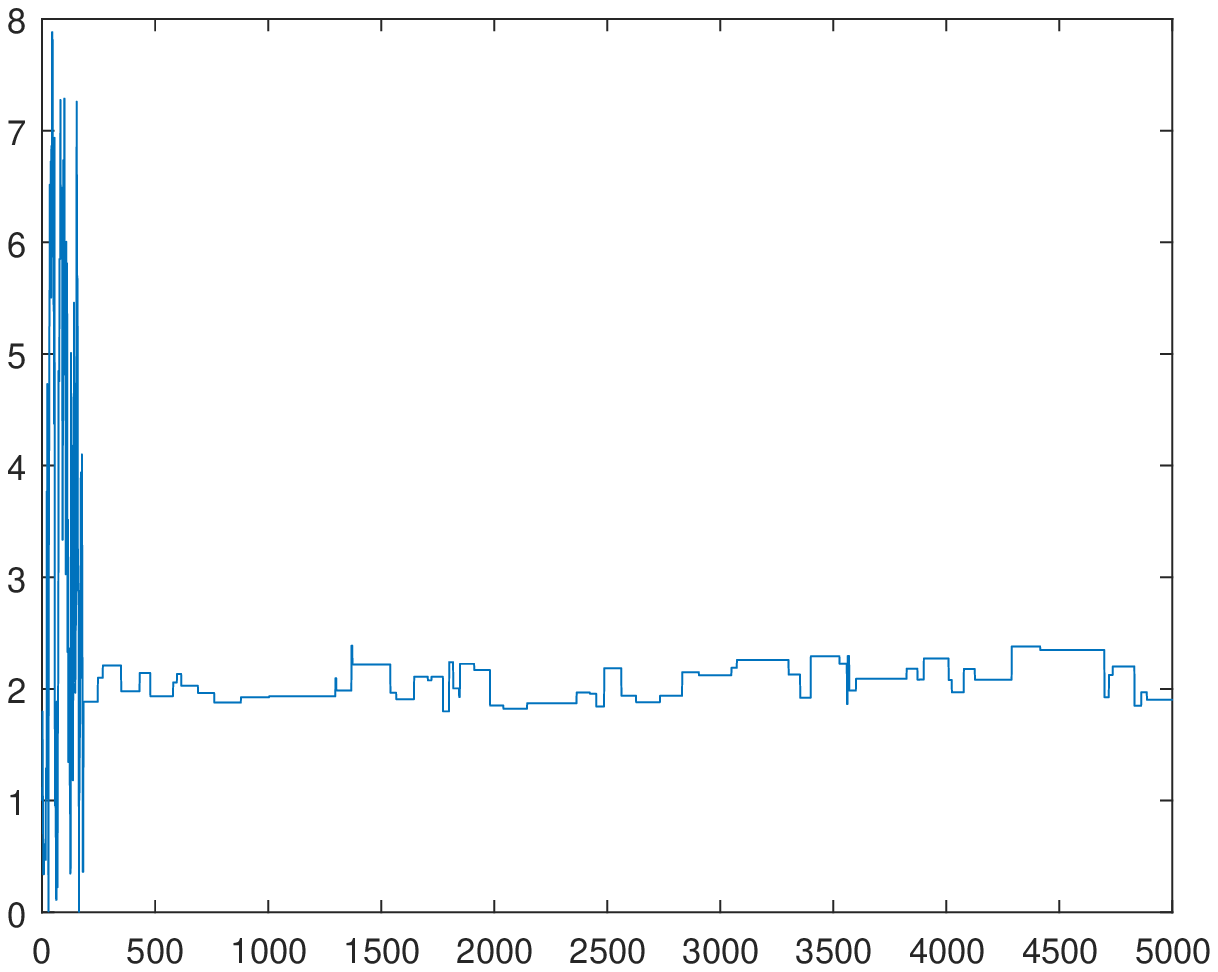}}
  \subfigure[\textbf{L-shaped (imag)}]{
    \includegraphics[width=2.5in]{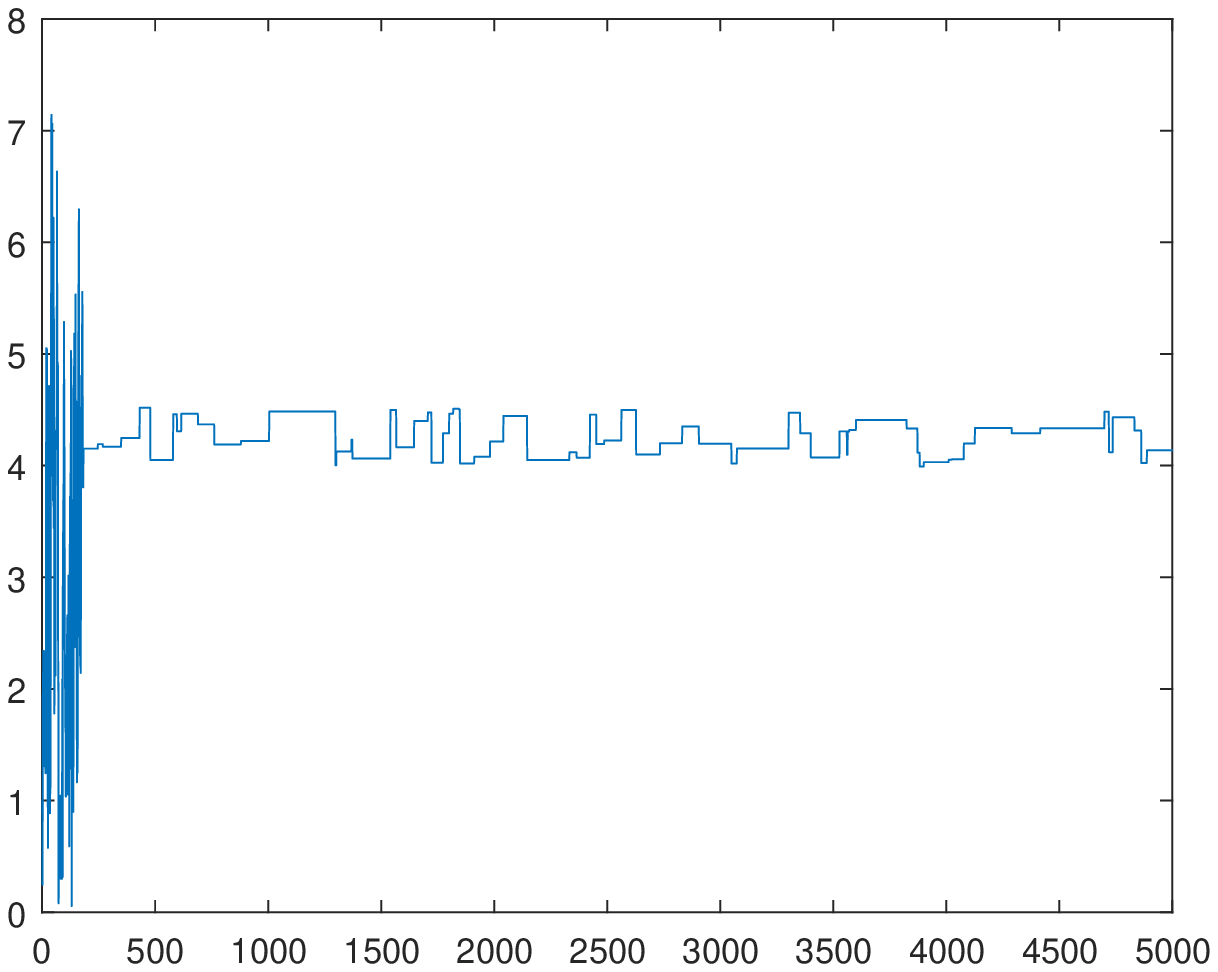}}
\caption{Trace plots of Markov chains for three domains ($n\equiv 2+4i$).}
\label{fn24i}
\end{figure}

\section{Conclusions}
In this paper, we show that the Cauchy data of a medium scattering problem can be used to reconstruct
the related Stekloff eigenvalues, which is demonstrated by numerical examples. A Bayesian approach is proposed
to estimate the index of refraction using the reconstructed Stkeloff eigenvalues. The method is particularly meaningful when there is a lack of
understanding of the relation between the known data and the unknown quantities. In future, the authors plan to extend the preliminary study here
to more challenging inverse scattering problems.

 \section*{Acknowledgement}{The research was supported in part by Guangdong Natural Science Foundation of China (2016A030313074) and NNSF of China under grants 11801218 and 11771068.}


\begin{thebibliography}{99}
\bibitem{AudibertCakoniHaddar2017} L. Audibert, F. Cakoni and H. Haddar,
	{\em New sets of eigenvalues in inverse scattering for inhomogeneous media and their determination from scattering data.}
	Inverse Problems 33, no. 12, 125011, 2017.
\bibitem{BaussardEtal2001IP} A. Baussard, D. Pr\'{e}mel, O. Venard, 
	{\em A Bayesian approach for solving inverse scattering from microwave laboratory-controlled data.} 
	Special section: Testing inversion algorithms against experimental data. Inverse Problems 17, no. 6, 1659-1669, 2001.
\bibitem{BondarenkoHarrisKleefeld2017AA} O. Bondarenko, I. Harris, and A. Kleefeld, 
	{\em The interior transmission eigenvalue problem for an inhomogeneous media with a conductive boundary.}
	 Appl. Anal. 96, no. 1, 2-22, 2017.
\bibitem{CakoniColtonHaddar2010CRASP} F. Cakoni, D. Colton and H. Haddar,
	{\em On the determination of Dirichlet or transmission eigenvalues from far field data.}
	CR Acad. Sci. Paris 348, 379-383, 2010.
\bibitem{CakoniColton2016} F. Cakoni, D. Colton, S. Meng and P. Monk,
	{\em Stekloff eigenvalues in inverse scattering.}
	SIAM J. Appl. Math. 76(4), 1737-1763, 2016.
\bibitem{CakoniColtonMonk2007IP} F. Cakoni, D. Colton and P. Monk,
	{\em On the use of transmission eigenvalues to estimate the index of refraction from far field data.}
	Inverse Problems 23, no.2, 507-522, 2007.
\bibitem{chenliu2005} Z. Chen and X. Liu,
	{\em An adaptive perfectly matched layer technique for time-harmonic scattering problems},
	SIAM J. Numer. Anal. 43(2), 645-671, 2005.
\bibitem{ColtonHaddar2005IP} D. Colton and H. Haddar,
	{\em An application of the reciprocity gap functional to inverse scattering theory.}
	Inverse Problems 21, no. 1, 383-398, 2005.
\bibitem{ColtonKress2013} D. Colton and R. Kress, Inverse Acoustic and Electromagnetic Scattering Theory,  3rd ed., Springer, 2013.
\bibitem{DiCristoSun2006IP} M. Di Cristo and J. Sun,
	{\em An inverse scattering problem for a partially coated buried obstacle.}
	Inverse Problems 22, no. 6, 2331-2350, 2006.
\bibitem{GharsalliEtal2014IP} L. Gharsalli, H. Ayasso, B. Duch\^{e}ne and A. Mohammad-Djafari, 
	{\em Inverse scattering in a Bayesian framework: application to microwave imaging for breast cancer detection.} 
	Inverse Problems 30(11), 114011, 2014.
\bibitem{GemanGeman1987} S. Geman and D.  Geman,
	{\em Stochastic relaxation, Gibbs distributions, and the Bayesian restoration of images}.
	IEEE Trans. Pattern Anal. Mach. Intell. PAMI-6, no. 6, 721-741, 1984.
\bibitem{HarrioEtal2006SC} H. Haario, M. Laine, A. Mira and E. Saksman, 
	{\em DRAM: Efficient adaptive MCMC}. 
	Stat. Comput. 16, 339-354, 2006.
\bibitem{HarrisCakoniSun2014IP} I. Harris, F. Cakoni and J. Sun, 
	{\em Transmission eigenvalues and non-destructive testing of anisotropic magnetic materials with voids.} 
	Inverse Problems 30, no. 3, 035016, 2014.
	
\bibitem{HarrisRome2017AA} I. Harris and S. Rome, 
	{\em Near field imaging of small isotropic and extended anisotropic scatterers.} 
	Appl. Anal. 96, no. 10, 1713-1736, 2017.
\bibitem{Hastings1970} W. Hastings,
	{\em Monte Carlo sampling methods using Markov chains and their applications.}
	Biometrika 57, no. 1, 97-109, 1970.
\bibitem{Huang2016JCP} R. Huang, A. Struthers, J. Sun and R. Zhang,
	{\em Recursive integral method for transmission eigenvalues.}
	 J. Comput. Phys. 327, 830-840, 2016.
\bibitem{Huang2017} R. Huang, J. Sun and C. Yang,
	{\em Recursive Integral Method with Cayley Transformation}.
	Numer. Linear Algebra Appl. 25, no. 6, 2018. https://doi.org/10.1002/nla.2199.
\bibitem{JariE.Somersalo2006} J. Kaipio and E. Somersalo, Statistical and Computational Inverse Problems. Springer, New York, 2006.
\bibitem{KirschLechleiter2013IP} A. Kirsch and A. Lechleiter,
	{\em The inside-outside duality for scattering problems by inhomogeneous media.}
	Inverse Problems 29, no.10, 104011, 2013.
\bibitem{LechleiterPeters2015CMS} A. Lechleiter and S. Peters,
	{\em Determining transmission eigenvalues of anisotropic inhomogeneous media from far field data.}
	Commun. Math. Sci.13, no. 7, 1803-1827, 2015.
\bibitem{LechleiterRennoch2015} A. Lechleiter and M. Rennoch,
	{\em Inside-outside duality and the determination of electromagnetic interior transmission eigenvalues.}
	SIAM J. Math. Anal. 47, 684-705, 2015.
\bibitem{LechleiterSun2017AA} A. Lechleiter and J. Sun (ed.), Special Issue on Recent Developments in Scattering and Inverse Scattering Problems,
	Appl. Anal.  96, no. 1, 1-172, 2017.
\bibitem{LiHuangLinWang2018IPI} T. Li, T.M. Huang, W.W. Lin, and J.N. Wang, 
	{\em On the transmission eigenvalue problem for the acoustic equation with a negative index of refraction and a practical numerical reconstruction method.}
	Inverse Probl. Imaging 12, no. 4, 1033-1054, 2018.
\bibitem{LiuSunTurner2018} J. Liu, J. Sun, and T. Turner,
	{\em Spectral indicator method for a non-selfadjoint Steklov eigenvalue problem.}
	J. Sci. Comput., https://doi.org/10.1007/s10915-019-00913-6, 2019.
\bibitem{LiuSun2014IP} X. Liu and J. Sun,
	{\em Reconstruction of Neumann eigenvalues and the support of a sound hard obstacle.}
	Inverse Problems 30, no. 6, 065011, 2014.
\bibitem{Metropolis1953} N. Metropolis, A.W. Rosenbluth, M.N.  Rosenbluth, A.H. Teller and E. Teller,
	{\em Equation of state calculations by fast computing machines.}
	 J. Chem. Phys.  21, no. 6, 1087-1092, 1953.
\bibitem{MonkSun2007IPI} P. Monk and J. Sun,
	{\em Inverse scattering using finite elements and gap reciprocity.} 
	Inverse Probl. Imaging 1, no. 4, 643-660, 2007.
\bibitem{PeterKleefeld2016IP} S. Peters and A. Kleefeld, 
	{\em Numerical computations of interior transmission eigenvalues for scattering objects with cavities.}
	Inverse Problems 32, no. 4, 045001, 2016.
%\bibitem{Peters2017AA} S. Peters,
%	{\em The inside-outside duality for elastic scattering problem.}
%	Appl. Anal. 96, no. 1, 48-69,  2017.
%\bibitem{PetersLechleiter2015IP} S. Peters and A. Lechleiter,
%	{\em The inside-outside duality for inverse scattering problems with near field data.}
%	Inverse Problems 31(8), 085004, 2015.
\bibitem{Stuart2010AN} A.M. Stuart
	{\em Inverse problems: a Bayesian perspective.}
	 Acta Numer. 19, 451-559, 2010.
\bibitem{Sun2011IP} J. Sun,
	{\em Estimation of transmission eigenvalues and the index of refraction from Cauchy data.}
	Inverse Problems 27, no. 1, 015009, 2011.
\bibitem{Sun2012IP} J. Sun, 
	{\em An eigenvalue method using multiple frequency data for inverse scattering problems.} 
	Inverse Problems 28, no. 2, 025012, 2012.
\bibitem{SunZhou2016} J. Sun and A. Zhou,
	Finite Element Methods for Eigenvalue Problems. CRC Press, Taylor $\&$ Francis Group, Boca Raton, London, New York, 2016.
\bibitem{YangMaZheng2015} Y. Wang, F. Ma and E. Zheng, 
	{\em Bayesian method for shape reconstruction in the inverse interior scattering problem.} 
	Math. Probl. Eng., Art. ID 935294, 2015.
\end{thebibliography}
\end{document}